\documentclass{article}

\usepackage{arxiv}

\usepackage[utf8]{inputenc} 
\usepackage[T1]{fontenc}    
\usepackage{graphicx}
\usepackage{hyperref}       
\usepackage{nicefrac}       
\usepackage{microtype}      
\usepackage[sort,numbers]{natbib}
\usepackage{doi}

\usepackage{amsmath}
\usepackage{array}
\usepackage{amssymb,amsthm}
\usepackage{enumitem}
\usepackage{multirow}
\usepackage{color}
\usepackage{colortbl}
\usepackage{tabularx,booktabs,makecell}
\usepackage{ragged2e}

\usepackage[tight,footnotesize]{subfigure}
\usepackage{setspace}
\usepackage{bm}
\usepackage{xspace}
\usepackage[ruled, vlined, linesnumbered]{algorithm2e}

\newtheorem{theorem}{Theorem}

\newtheorem{corollary}{Corollary}

\newcommand{\rf}{{\tilde{\rho}}\xspace} 
\newcommand{\while}{{\tt whileloop}\xspace}
\newcommand{\for}{{\tt forloop}\xspace}
\newcommand{\cmin}{{c^{\rm min}}\xspace} 
\newcommand{\bgt}{{\bf{c}}\xspace} 
\newcommand{\unst}{\mbox{$\mathsf{UNS}_{\alpha}$}\xspace} 
\newcommand{\bs}{\backslash}
\newcommand{\dsys}{{\mathcal{I}}\xspace} 
\newcommand{\bsRand}{{\sf rand}\xspace} 
\newcommand{\bsRMax}{{\sf rMax}\xspace} 

\newcommand{\mytitle}{\LARGE Nonmonontone Submodular Maximization under Routing Constraints \thanks{\{cs\_haotian, liraorr57, sungd\}@bjfu.edu.cn, zewei.wu@ipm.edu.mo}}

\title{\mytitle}
\date{December 2022}

\author{
{Haotian Zhang} 
\And
{Rao Li} 
\And
{Zewei Wu} 
\And
{Guodong Sun} 
}

\renewcommand{\shorttitle}{}

\hypersetup{
pdftitle={A template for the arxiv style},
pdfsubject={q-bio.NC, q-bio.QM},
pdfauthor={David S.~Hippocampus, Elias D.~Striatum},
pdfkeywords={First keyword, Second keyword, More},
}

\begin{document}
\maketitle

\begin{abstract}
In machine learning and big data, the optimization objectives based on set-cover, entropy, diversity, influence, feature selection, etc. are commonly modeled as submodular functions.  Submodular (function) maximization is generally NP-hard, even in the absence of constraints. 
Recently, submodular maximization has been successfully investigated for the settings where the objective function is monotone or the constraint is computation-tractable. However, maximizing nonmonotone submodular function with complex constraints is not yet well-understood. 
In this paper, we consider the nonmonotone submodular maximization with a cost budget or feasibility constraint (particularly from route planning) that is generally NP-hard to evaluate. This is a very common issue in machine learning, big data, and robotics. This problem is NP-hard, and on top of that, its constraint evaluation is likewise NP-hard, which adds an additional layer of complexity. So far, few studies have been devoted to proposing effective solutions, leaving this problem currently unclear.
In this paper, we first present an iterated greedy algorithm, which offers an approximate solution. Then we develop the proof machinery to demonstrate that our algorithm is a bicriterion approximation algorithm: it can accomplish a constant-factor approximation to the optimal algorithm, while keeping the over-budget tightly bounded. We also look at practical concerns for  striking a balance between time complexity and over-budget. Finally, we conduct numeric experiments on two concrete examples to show our design's efficacy in real-world scenarios. 
\end{abstract}


\section{Introduction}\label{sec:introduction}

Submodularity fundamentally captures a diminishing-return property of set function:  the reward of adding an item into a set will decrease as this set grows. Submodular maximization aims at finding a subset of some ground set such that the submodular objective function can be maximized with some (or without any) constraints on cost or feasibility. Submodular function maximization appears in a wide variety of scenarios, from combinatorial optimization to machine learning, such as the coverage problem, facility location planning, information gathering, feature selection, data summarization, factorization in graphical models, and so forth~\citep{Krause-2012-submodular-surv,Wu-2018-submodular,Bilmes-2022-submodular-surv}. 

In general, submodular maximization problems are NP-hard, regardless of whether they are constrained or not. A lot of  approximation algorithms with guaranteed performance have been presented to maximize (1) unconstrained functions or (2) constrained-but-monotone functions. 
So far, however, the \emph{constrained nonmonotone submodular maximization} has not been fairly well-understood. 
Recently, a number of studies have tried to maximize nonmonotone submodular functions with additive or cardinality-based constraints, and presented algorithms with provable guarantees. 
In practical settings, however, the constraints are often more complex, or even infeasible to evaluate in poly-time~\citep{Wu-2019-wpt-mobile,Amantidis-2019-submodular,Jakkala-2022-submodular-app,Durr-2020-submodular}. For instance, in the big-data system with networked processors, multiple  processors are often selected to concurrently execute a computation task on their local datasets. 
For timely applications, these selected processors are orchestrated as a delay-minimum multicast routing tree and expected to return to the user more representative results within a time threshold. 
Another example of nonmonotone submodular maximization is to complete an informative data collection with an energy-budgeted robot, which passes through a TSP route on some points of interest to collect data.
These two examples will be detailed in section \ref{sub:problem}. Their objective functions are nonmonotone submodular, and their constraint evaluation is a cost-minimization problem, which is typically NP-Hard and more complex than evaluating additive or cardinality-based constraints. 

In this paper, our goal is to maximize nonmonotone submodular function with the constraints that need to be evaluated by some route planning. 
To the best of our knowledge, no literature has investigated such a submodular-maximization setting constrained by a routing budget.  
Part of our problem's difficulty stems from its submodularity, but a great deal more of the difficulty is attributed to its computationally-intractable routing constraints, which hinder the development of constant-factor approximate solution. 
In response to this issue, we present an iterated two-stage algorithm framework, which uses any poly-time, $(1+\theta)$-approximation algorithm to evaluate whether the routing cost stays under budget. We elaborately draw a connection between the error parameter $\theta$ and our algorithm's overall performance. Additionally, by exploring the routing cost function, we prove that it belongs to a kind of $k$-independence system, where $k$ is loosely upper-bounded by the ground set's size and usually a positive constant in practical settings. Such exploration offers our algorithm a lower bound for performance guarantee. Based on all the above, we develop the proof machinery that proves our algorithm can achieve a $[1/4k, 1+\theta]$-bicriterion approximation solution; that is, our objective value is at least as good as $1/4k$ of the optimal one, while the ratio of over-budget is upper bounded by $\theta$ (which is usually less than one or can be any small in some cases).

The remainder of this paper is organized as follows.  Section \ref{sec:RelatedWork} discusses the related works. Section \ref{sec:ModelProblem} gives the necessary preliminaries and formulates our problem. Section \ref{sec:Designs} details the designs of our algorithm and analyze its performance bound. Section \ref{sec:Experiment} numerically evaluates our designs with two case studies. Finally, section \ref{sec:Conclustion} concludes this paper.

\section{Related Work}\label{sec:RelatedWork}

Submodular function maximization subsumes a wide range of NP-Hard combinatorial optimization problems, and then, has long been studied~\cite{Krause-2012-submodular-surv}. The unconstrained submodular maximization is NP-hard but well-approximable. For the unconstrained settings, if the objective function is nonnegative and monotone, the simple greedy policy that iteratively selects the item with the highest marginal profit can achieve a poly-time, $(1-1/e)$-approximation solution. In practice, maximizing submodular function is usually restricted by  some feasibility or cost constraints. In this section, therefore, we will only give a brief introduction to some main studies on the constrained submodular maximization and their results, some of which are listed in Table \ref{tab:ratios}.

\begin{table}[t]
\renewcommand\arraystretch{1.15}	
\centering
\caption{Main results of submodular maximization}\label{tab:ratios}
\small
\begin{tabularx}{\linewidth}{lcrX}\toprule
paper (year) & \thead{objective function} & \thead{constraint class} & \thead{approximation ratio} \\ \midrule
\cite{Nemhauser-1978-submodular} (1978) & monotone  & cardinality & $1-1/e\approx 0.632$  \\
\citep{Krause-2005-submodular} (2005) & monotone & cardinality & randomized algo.: $1-1/e-\varepsilon$, $\varepsilon>0$ is any small \\
\citep{Sviridenko-2004-submodular} (2004)  & monotone & 1-knapsack & $1-1/e$  \\
\citep{Calinescu-2011-submodular-matroid} (2011)  & monotone & $1$-matroid & $1-1/e$ \\
\citep{Badanidiyuru-2014-submodular} (2014)  & monotone  & $k$-system $l$-knapsack & $1/(k+2l+1+\varepsilon)$, $\varepsilon>0$ determines time complexity \\
\citep{Zhang-2016-submodular-rt} (2016) &  monotone & $\alpha$-submodular, routing & $[\frac{1}{2}(1-1/e), p(\alpha, k_c,\psi)]$, bi-criterion   \\ 
\citep{Wu-2019-wpt-mobile} (2019) & monotone & TSP routing &  $[\frac{1}{4}(1-1/e), p(\alpha, k_c,\psi)]$, bi-criterion   \\ 
\citep{Buchbinder-2012-nonmonotone-submodular} (2012) & nonmonotone & -- & deterministic algo.: 1/2; randomized algo.: 2/5\\
\citep{Feldman-2011-submodular-matroid} (2011) & nonmonotone & 1-matroid & $1/e\approx 0.368$  \\
\citep{Buchbinder-2014-nonmonotone-submodular} (2014) & nonmonotone & cardinality & $1/e+0.004\sim 1/2$  \\
\citep{Lee-2009-submodular} (2009) & nonmonotone & $k$-matroid & $\frac{1}{k+2+1/k+\epsilon}$, $\epsilon>0$ is any small  \\
\citep{AAGK-2010-nonsubmodular} (2010) & nonmonotone & $k$-system & $\frac{k}{(1+1/\alpha)(k+1)^2}\approx \frac{1}{(1+1/\alpha)k}$, $0<\alpha<1$ \\
\citep{Mirzasoleiman-2016-submodular} (2016) & nonmonotone & $k$-system & $\frac{k}{(k+1)(2k+1)}\approx\frac{1}{2k}$ \\
\citep{Feldman-2017-nonmonotone-submodular} (2017) & nonmonotone & $k$-system, $k$-extendible & $\frac{1}{k+\mathcal{O}(\sqrt{k})}$ \\
\citep{Shi-2020-submodular-psystem} (2020) & nonmonotone & $k$-system & $\frac{1}{2k+3+1/k}$ \\
\citep{Tang-2022-submodular} (2022) & nonmonotone & $k$-system & randomized algo.:  $\frac{1}{2k+4}$ for adaptive submodularity\\
\bottomrule
\end{tabularx}
\end{table}

\subsection{Monotone Submodular Maximization}

In literature, most of efforts are put on maximizing monotone submodular functions with various types of constraints, and constant-factor approximation algorithms are presented. 

\textbf{Cardinality Constraint}. \citet{Nemhauser-1978-submodular} first give the near-optimal solution for the cardinality setting, with an approximation ratio of $(1-1/e)$. 
\citet{Krause-2005-submodular} give a randomized algorithm to achieve $(1-1/e+\varepsilon)$-approximation with high probability (here $\varepsilon>0$ can be any small). 
\citet{Balkanski-2018-submodular}  study the adaptive complexity of maximizing a submodular function with a cardinality constraint, and prove that there exists $\mathcal{O}(\log n)$-adaptive algorithm to achieve a $(1/3-\varepsilon)$-approximation for any small $\varepsilon>0$.
In \cite{Breuer-2020-submodular-cardi}, the authors introduce a parallel algorithm, and the performance is arbitrarily close to $(1-1/e)$. 

\textbf{Matroid or Knapsack Constraint}. \citet{Calinescu-2011-submodular-matroid} give a $(1-1/e)$-approximation algorithm for monotone submodular maximization subject to a general matroid constraint.
For the settings with $k$-matroid constraints, \citet{Nemhauser-1978-submodular} proposes a $1/(k+1)$-approximation algorithm. 
\citet{Badanidiyuru-2014-submodular} use the continuous greedy policy to design a $1/(k+2l+1+\varepsilon)$-approximation for the intersection of a $k$-system and linear $l$-knapsack  constraints. 

\textbf{More Complex Constraint}.  A more recent work~\cite{Zhang-2016-submodular-rt} considers the monotone submodular maximization with an $\alpha$-submodular constraint (such as routing constraint), where evaluating the constraint is assumed to be NP-hard. The authors give a bi-criterion algorithm (i.e., their solution is bounded in performance, though it will to some extent disagree with the original constraint), with an approximation ratio of $(1-1/e)/2$; the constraint violence is bounded by a polynomial $p(\alpha, k_c, \psi)$, where $k_c$ is the curvature of the constraint function and $\psi$, the approximation ratio of the algorithm used to evaluate the constraint. 
\citet{Wu-2019-wpt-mobile} model their wireless charging scheduling problem as a monotone submodular maximization problem restricted by TSP routing; they replace their objective function with a surrogate function and propose a $(1-1/e)/4$-approximation bi-criterion algorithm. 
Nevertheless, these two studies' objective functions are monotone, and their analytical machinery cannot be applied to nonmonotone cases at all. 

\subsection{nonmonotone Submodular Maximization} 

Maximizing nonmonotone submodular functions is generally more difficult and usually investigated  in the cases without any constraints, or just with cardinality, convex-set, or packing-type constraints. 
For the unconstrained nonmonotone cases, \citet{Buchbinder-2012-nonmonotone-submodular} propose a local search-based exact algorithm and a randomized algorithm, both of which can achieve approximation ratios of $1/2$ and $2/5$, respectively. Next we give main results of maximizing constrained nonmonotone submodular functions.

\textbf{Cardinality Constraint}. \citet{Buchbinder-2014-nonmonotone-submodular} present an algorithm with constant factor ranging from $(1/e + 0.004)$ to 1/2, and a 0.356-approximation algorithm under an exact cardinality constraint. 
\citet{Wang-submodular-partition} investigate the submodular partitioning problem with cardinality constraint, which divides a ground set into $m$ blocks to maximize the evaluation of the minimum block under constraints. They propose a greedy-based algorithmic framework to achieve an $\Omega(1/m)$-approximation. 

\textbf{Matroid or Knapsack Constraint}. In the case with a single matroid constraint, \citet{Feldman-2011-submodular-matroid} design a $1/e$-approximation greedy algorithm.
\citet{Lee-2009-submodular} give the first constant-factor results for nonmonotone settings under $k$-matroid or $k$-knapsack constraints; the approximation ratios are $1/(k+2+1/k+\epsilon)$ and $(0.2-\epsilon)$ for the $k$-matroid and $k$-knapsack cases, respectively, where $\epsilon>0$ is any small.

\textbf{$k$-independence System Constraint}. \citet{AAGK-2010-nonsubmodular} investigate the algorithmic framework of nonmonotone submodular maximization constrained by a $k$-independence system ($k$-system, in short), which is a generalization of intersection of $k$ matroids. Their algorithm achieves an approximation ratio of $\frac{k}{(1+1/\alpha)(k+1)^2}\approx \frac{1}{(1+1/\alpha)k}$,  where $0<\alpha<1$ is the approximation guarantee for unconstrained nonmonotone submodular maximization. This result has remained competitive until now, and has motivated a number of studies focusing on the nonmonotone settings with $k$-system constraint. 
\citet{Mirzasoleiman-2016-submodular} follow the above algorithm and slightly reduce the above approximation ratio down to $\frac{k}{(k+1)(2k+1)}$. 
Also adopting the iterative algorithmic framework of~\cite{AAGK-2010-nonsubmodular}, \citet{Feldman-2017-nonmonotone-submodular} propose a deterministic algorithm that achieves an approximation ratio of $\frac{1}{k+\mathcal{O}(\sqrt{k})}$, and furthermore, they propose a randomized $\frac{k}{(k+1)^2}$-approximation algorithm for the nonmonotone settings with $k$-extensible constraint. 
\citet{Shi-2020-submodular-psystem} extend the approach in \cite{Lee-2009-submodular} to make it suitable to the $k$-system constraint, and present an algorithm with the factor of $\frac{1}{2k+3+1/k}$; the authors also show that, if $k<8$, their algorithm will outperform that of \cite{Feldman-2017-nonmonotone-submodular}. 
\citet{Tang-2022-submodular} proposes a sampling-based randomized algorithm for maximizing the $k$-system-constrained nonmonotone adaptive submodular function, and this algorithm achieves an approximation ratio of $1/(2k+4)$.

To date, the algorithms for nonmonotone submodular maximization under $k$-system have almost followed the seminal framework proposed in~\cite{AAGK-2010-nonsubmodular}. All these algorithms are $\frac{1}{\Omega(k)}$-approximation, merely with slightly-different time complexities. 
In particular, their constraint functions are assumed to be computationally tractable. However,  
it is not always true; for instance, evaluating routing restriction is generally NP-Hard. 
Until now, there have not been any algorithmic frameworks and analytical machinery that can achieve a constant-factor or bi-criterion approximation for the nonmonotone submodular maximization with an intractable routing constraint. 

\section{Preliminaries and Problem Statement}\label{sec:ModelProblem}

In this section, we first introduce some preliminaries as well as notations, and then, formalize our problem along with two motivating examples. Finally we analyze the intractability of our problem. 

\subsection{Submodular Maximization}

Given the ground set $\Omega$ of $n$ items, we say that a utility function $f: 2^\Omega\rightarrow\mathbf{R}$ is a \emph{submodular set function} if and only if we have 
\begin{equation}\label{eqn:submodularity1}
f(A\cup B) + f(A\cap B) \leq f(A) + f(B)
\end{equation}

\noindent for any $A\subseteq\Omega$ and $B\subseteq\Omega$~\cite{Fujishige-2005-submodular-bk}. Commonly, $f$ is nonnegative and $f(\emptyset)=0$. The submodularity of $f$ can be equivalently expressed with 
\begin{equation}\label{eqn:submodularity2}
f(A\cup\{x\}) - f(A) \geq f(B\cup\{x\}) - f(B)
\end{equation}

\noindent for any $A\subseteq B\subset \Omega$ and $x\in \Omega\bs B$. Submodular function $f(A)$ has a diminishing marginal gain as $A$ grows. We hereafter use $f^+_A(x) \triangleq  f(A\cup\{x\}) - f(A)$ to denote the marginal gain of $f$ on $x\in\Omega$ with respect to $A\subset\Omega$, unless otherwise noted. 
A submodular function $f$ defined on $2^\Omega$ is monotone if for all $A\subseteq B\subseteq \Omega$, we always have $f(A)\leq f(B)$ or a nonnegative marginal gain. 
In a nonmonotone submodular case,  however, $f^+_A(x)$ could be negative.

A constrained submodular maximization problem can be defined by $\max\{f(S)|g(S)\leq c,\,S\subseteq\Omega\}$, where $f$ is the objective function, and $g$, the cost function. In practice, cost functions are usually positive and monotone increasing. 

\subsection{Routing-based Cost Function}

Given a ground set $\Omega$ of items, a weighted graph abstraction can be drawn over $\Omega$---the graph vertices are the items, and an edge exists between two items if both of them have a connection in some way. In such a graph, vertices (items) and edges are all associated with positive weights to represent the costs of visiting vertices and passing through edges. 

With a nonnegative function $\rho: 2^\Omega\rightarrow\bf{R}_{\geq 0}$ and any $S\subseteq\Omega$, if $\rho(S)$ is mapped to the least total cost of the route that visits all the vertices of $S$ in some way, we say $\rho$ is a \emph{routing-based cost function}.  
More specifically, the routing-based cost function on $S$ can be commonly written as $\rho(S)=\sum_{s\in S}c_s+r^*(S)$, where $c_s\geq 0$ is the cost in visiting $s$ and $r^*(S)$ is the minimum cost of traversing each vertex of $S$ at least once. In general, $r^*(S)$ is infeasible to compute within poly-time; clearly, it is the $r^*(S)$ term that results in the NP-hardness of determining $\rho(S)$.

To simplify notation, we also use $\rho$ to represent the optimal algorithm of achieving the least cost. In most real-life applications, we have $\rho(T)\leq\rho(S)$ for $T\subseteq S$, i.e., $\rho$ is nondecreasing. Specifically, it is easy to prove that for any $S, T\subseteq\Omega$ with $S=T\cup\{s\}$, we have $\rho(S)-\rho(T)\geq \cmin$ where $\cmin=\min\{c_s|s\in\Omega\}$, i.e., the gradient of function $\rho$ is at least $\cmin$. In this paper we assume $\rho(\emptyset)=0$ and $\rho$ is nondecreasing. 

\subsection{Problem Definition and Examples}\label{sub:problem}

Given a nonempty ground set $\Omega$ and a nonmonotone submodular set function $f: 2^\Omega\rightarrow\mathbf{R}_{\geq 0}$, we aim at finding a subset $S^*\subseteq\Omega$ such that 
\begin{equation}
S^* = \arg\max\limits_{S\subseteq\Omega}\{f(S)~|~\rho(S)\leq\bf{c}\}\,,
\end{equation}

\noindent where $\rho(S)$ is a routing-based cost function defined on $2^\Omega$ and $\bf{c}$ is the budget (a positive constant). Here, $\rho(\cdot)\leq\bf{c}$ is a \emph{routing constraint} to the objective function $f(\cdot)$. We assume that $\rho(\{s\})\leq\bgt$ for any singleton $s\in\Omega$. 
Followed are two motivating examples of our problem. 

\textbf{Example-1: Diversified Task Offloading in Big Data}. Big Data system is a networked distributed computing environment, in which each processor can store and process data, and processors are connected with delaying links. To effectively complete a computation task within a given time, users usually offload the task down to a subset of processors whose data is more representative. These selected processors can, in parallel, run the task on their local data sets, and return results to users for further combination. 
In machine learning~\cite{Lin-2011-diversity,Mirzasoleiman-2016-submodular,Mirzasoleiman-2018-nonmonotone-submodular-app}, the functions of evaluating data representativeness are commonly nonmonotone and submodular (see Appendix \ref{appdxB:similarity}). 
The user can calculate (foreknow) the time of each processor running a task before offloading it. Given a subset of processors, therefore, the completion time heavily depends on how long the in-network data transfer will take. Equivalently, what the user needs is find a delay-minimum multicast routing tree that is rooted at the user and spans all the selected processors, which is an NP-Complete task in general. Obviously, this example can be modeled as a nonmonotone submodular maximization with a multicast routing restriction. 

\textbf{Example-2: Informative Data Collection with a Robot}.  In this example, an on-ground or aerial robot departs from a depot and moves through some points of interest (already specified in a certain monitoring area) to capture data. The robot's objective is collect as informative data as possible, while not running out of its energy before returning to the depot. This scenario often falls into the category of informative path planning in robotics or social networks. The robot cannot visit all points due to its limited energy capacity. It has to select an informative subset of points, and visit them along a TSP tour that starts and ends at the depot. In real-life applications~\cite{Sharma-2015-entropy,Bachman-2019-mi,Tschannen-2020-representation}, ``informative'' can be measured with the mutual information between the visited and unvisited points; and mutual information is of nonmonotone submodularity. So, this example can be naturally modeled as a nonmonotone submodular maximization with TSP constraint.

\subsection{Hardness of Our Problem}

The computational intractability of our problem stems from three aspects. First, even if we ignore the constraint on our problem (for example, $\bf{c}$ is very large), maximizing a (nonmonotone) submodular function remains NP-hard. It is quite straightforward to know that our problem is NP-hard, too. Second, for a given subset of $\Omega$, we need to evaluate its feasibility by solving the routing-based cost function that is NP-hard in general settings. Third, our problem's objective function and constraint evaluation cannot be  solved separately, because each is coupled with the other.

There have been lots of efforts to optimize monotone and nonmonotone submodular functions without constraints. In the past few years, however, there has been a surge of interest in applying constrained submodular  maximization to machine learning. But most of these prior works usually focus on the maximization with computationally tractable constraints, such as cardinality, knapsack, and matroid constraints~\cite{Lee-2009-submodular,Iyer-2013-submodular,Buchbinder-2014-nonmonotone-submodular}. They, however, cannot be directly applied to our problem to achieve bounded performance.

\section{Our Algorithm}\label{sec:Designs}

In the seminal work~\cite{AAGK-2010-nonsubmodular}, the authors propose an iterative greedy algorithm (called SMS, for convenience) to maximize nonmonotone submodular function with $k$-system constraints. SMS can provide an approximation ratio of $\frac{k}{(1+1/\alpha)(k+1)^2}\approx \frac{1}{(1+1/\alpha)k}$, which remains competitive for general settings until now. Here, parameter $\alpha~(0<\alpha<1)$ is the approximation guarantee for unconstrained nonmonotone submodular maximization. 
Each iteration of SMS covers two stages: (1) selecting as many items from available ones as possible, until the $k$-system property cannot hold, and (2) applying an unconstrained nonmonotone submodular maximization algorithm to those selected items. SMS repeats the two stages $k$ times. During iterations, it stores all the feasible solutions determined by the second stage, and the best one will be returned at last. 

Our algorithm follows the algorithmic framework of SMS, and it can achieve at least the same performance as SMS's, while only needing polynomial time cost. Our basic idea is as follows. First, we prove the routing constraint on our problem is a kind of $k$-system, which transforms our problem into a $k$-system-constrained nonmonotone submodular maximization problem. This enables SMS's approximation ratio to work for our algorithm. Second, we recruit an algorithm $\rf$ to approximate $\rho$ in evaluating the routing cost, such that our algorithm can terminate in poly-time with an over-budget. Third, we prove that the amount of over-budget is bounded, by analyzing the performance gap in cost evaluation between $\rf$ and $\rho$. 

\subsection{Algorithm Description}

Our algorithm is outlined in Algorithm \ref{algo}, and it is an iterative procedure. Each iteration includes two stages which are described below in detail.

\let\oldnl\nl 
\newcommand{\nonl}{\renewcommand{\nl}{\let\nl\oldnl}}
\SetAlgorithmName{Algorithm}{}{}
\begin{algorithm}
\caption{our algorithm}\label{algo}
\DontPrintSemicolon
\SetKwInOut{Input}{input}\SetKwInOut{Output}{output}
\SetKwFor{While}{while}{}{end}
\SetKwFor{For}{for}{}{end}
\Input{$\Omega$, $\bf{c}$, $k$, $\theta$, $\alpha$}
\Output{$T^*\subseteq\Omega$ and $f(T^*)$}
\BlankLine
$S_1\leftarrow\Omega$ and $\mathcal{T}\leftarrow\emptyset$ {~~{\tcp*[h]{two sets}}}\;
\For{$i=1$ \emph{up to} $k$}{ \label{algoline:forloop}
  	$X_i \leftarrow \emptyset$ and $Y_i\leftarrow \emptyset$ {~~{\tcp*[h]{two sequences}}}\;
	{\nonl {\color{blue}{$\triangleright$ Stage 1: add as many items as possible}} }\;
	\While{$S_i\neq\emptyset$}{ \label{algoline:whileloop} 
	  	$s^* \leftarrow\arg\max\{f^+_{X_i}(s)|s\in S_i\}$ \label{algo:cherryPick}\;
		\eIf{$s^*$~\emph{exists and}~$\rf(X_i+s^*)\leq(1+\theta)\bgt$}{\label{algo:rtEvaluation}
			$X_i \leftarrow X_i+s^*$ and $S_i\leftarrow S_i\bs s^*$ \;
			\If{$\rf(X_i)>\bgt$}{\label{algo:beyondBudget}
				$Y_i\leftarrow Y_i + s^*$ \; \label{algo:Y}
			}
		}{
			break the while-loop (line \ref{algoline:whileloop})\;
		}
	}
	\BlankLine
	{\nonl {\color{blue}{$\triangleright$ Stage 2: find the best from a set of local optima}} }\;
	\eIf{$X_i\neq\emptyset$}{
		$T^*_i \leftarrow \unst(X_i)$ \label{algo:xISopt}\;
		\While{$y\leftarrow$~\emph{the last item of}~$Y_i$~\emph{exists}}{\label{algo:last}
			$X_i\leftarrow X_i\bs y$ and $Y_i\leftarrow Y_i\bs y$ \label{algo:curtailXi1}\;
		    $T_y \leftarrow \unst(X_i)$ and $\mathcal{T}\leftarrow\mathcal{T}\cup\{T_y\}$ \label{algo:curtailXi2}\;
                
		}
		$\mathcal{T}\leftarrow\mathcal{T}\cup\{T^*_i\}$ \;
	}{
		break the for-loop (line \ref{algoline:forloop})\;
	}
	$S_{i+1} \leftarrow S_i$ {~~{\tcp*[h]{the $i$-th iteration stops here}}}\; 
}
\Return $f(T^*)=\max\{f(T)|T\in\mathcal{T}\}$ and $T^*$\;
\end{algorithm}

\subsubsection{Stage-1: Maximization with Constraint Relaxed}

In the $i$-th iteration, the items, yet unexamined so far, form the set $S_i$. Stage-1 uses the \while iteratively to examine all the items of $S_i$, picking out $s^*$, the item bringing the highest marginal gain. During the cherry-picking, our algorithm uses a $(1+\theta)$-approximation poly-time algorithm $\rf$ in place of optimal algorithm $\rho$ to verify whether the routing cost is under constraint, where $\theta$ is usually within $(0,1)$.  

In the \while, the constraint is relaxed from $\bgt$ to $(1+\theta)\bgt$, and the \while continues until the routing cost evaluated by $\rf$ is beyond $(1+\theta)\bf{c}$. If the cost stays under the relaxed budget, $s^*$ will be moved from $S_i$ into the rear of $X_i$, which is an ordered set or a sequence; otherwise, the \while breaks, and then, Stage-1 is over in the current iteration and outputs $X_i$ and $Y_i$. As shown on lines \ref{algo:beyondBudget} and \ref {algo:Y}, $Y_i$ is also a sequence and only stores the items that are examined after the original budget $\bgt$ is violated. 
The \while tries to maximize $f$ with its original constraint slightly relaxed, and therefore, $\rf(X_i)$ output by Stage-1 may be larger than $\bgt$. But the over-budget can remain under $\theta\cdot\bgt$.

\subsubsection{Stage-2: Maximization without Constraints}

Once Stage-1 stops with $X_i\neq\emptyset$ as output, Stage-2 is prompt to take over the unfinished task. 
Stage-2 tries to employ the \unst procedure to dig out a more profitable subset out of $X_i$. Here, \unst can be any solvers that maximize a general unconstrained nonmonotone submodular function, with a bounded approximation ratio $\alpha$ ($0<\alpha<1$). In this paper, \unst is fulfilled by the deterministic algorithm proposed in \cite{Buchbinder-2012-nonmonotone-submodular}, which can use $\mathcal{O}(n)$ value oracles to achieve $1/3$-approximation for unconstrained nonmonotone submodular maximization problem. Appendix \ref{appdxA:deterministicUSM} shows the implementation. 

Stage-2 first uses \unst to pick out $T^*_i$ from $X_i$. However, $T^*_i$ is not necessarily the best because of $f$'s non-monotonicity. 
As shown on lines \ref{algo:curtailXi1} and \ref{algo:curtailXi2}, therefore, Stage-2 continues to curtail $X_i$ until $X_i\bs Y_i$ is reached, and in the meantime, it repeatedly recruits \unst to explore for the solution $T_y$ that is possibly better than $T^*_i$. 

After the \for completes at $i=k$ or $X_i=\emptyset$,  our algorithm will immediately scan through all the routes in $\mathcal{T}$ and then return the ultimate winner and its objective value. 
The \for goes ahead as $i$ increases from 1 to $k$. Here, $k$ is an integer constant, loosely upper-bounded by $|\Omega|-1$, and it will be discussed later in section \ref{subsec:timeComplexity}.

\subsection{Theoretical Performance}\label{subsec:approx}

Given a  set function $f: 2^\Omega\rightarrow\mathbf{R}_{\geq 0}$, we say $(\Omega, \mathcal{I}\subseteq 2^\Omega)$ is an \emph{independence system} if the following two properties are satisfied: I (non-emptiness) $\emptyset\in\mathcal{I}$, i.e., $\mathcal{I}$ is not empty per se, and II (heredity) if $S_1\in\mathcal{I}$ and $S_2\subset S_1$, then $S_2\in\mathcal{I}$.
The heredity of independence system makes $\dsys$ exponential in size in the worst case, and therefore, the optimization via exhaustive search over $\dsys$ is infeasible in terms of computation. 

Given $S\subseteq\Omega$ and $S_1\subseteq S$, we say $S_1$ is a {\it base} of $S$, if $S_1\in\dsys$ and we cannot find other $S_2\in \dsys$ such that $S_1\subseteq S_2\subseteq S$. 
An independence system $(\Omega, \dsys)$ is called a $k$-{\it independence system} ($k$-{\it system}, in short), if and only if there exists an integer $k$ such that ${|B_1|}/{|B_2|}\leq k$ for any two bases, $B_1$ and $B_2$, of $\Omega$. As a special case, a matroid is a 1-independence system. 

\begin{theorem}\label{thr:kSystem}
	\emph{
		Let $\dsys_\rho=\{S|S\subseteq\Omega~\mbox{and}~\rho(S) \leq \bf{c}\}$ and 
		$\dsys_\rf=\{S|S\subseteq\Omega~\mbox{and}~\rf(S)\leq (1+\theta)\mathbf{c}~
		\mbox{with}~\theta>0\}$. Both $(\Omega, \dsys_\rho)$ and $(\Omega, \dsys_\rf)$ 
		are a kind of $k$-system. 
	}
\end{theorem}
\begin{proof}
We first prove that $\dsys_\rho$ is an independence system. 	Since $\emptyset\subset\Omega$ and $\rho(\emptyset)=0$, we have $\emptyset\in\dsys_\rho$. Consider $S\in\dsys_\rho$ and a subset $T\subset S$. We know $\rho(T)<\rho(S)\leq\mathbf{c}$, which means $T\in\dsys_\rho$. Thus $(\Omega, \dsys_\rho)$ is an independence system. 

For any item $s\in\Omega$, we have assumed $\rho(\{s\})\leq\bgt$, indicating that each base of $\dsys_\rho$ is not less than one in size. On the other hand, it is easy to prove, by contradiction, that if $\Omega$ has a base of size $n$, then all its bases are of size $n$. 
That said, there must exist an integer $k_\rho$ ($1\leq k_\rho<n$) such that ${|B_1|}/{|B_2|}\leq k$ for any two bases, $B_1$ and $B_2$, of $\Omega$. Hence $\dsys_\rho$ is a $k_\rho$-system.
Applying similar analysis to $(\Omega, \dsys_\rf)$, we can conclude that there exists an integer $k_\rf$ ($1\leq k_\rf<n$) such that $(\Omega, \dsys_\rf)$ is a $k_\rf$-system.
\end{proof}

Theorem \ref{thr:kSystem} implies that our problem is also a nonmonotone submodular maximization problem with $k$-system constraint. That said, SMS's approximation ratio is hopefully true for our algorithm. 
However, SMS requires that its $k$-system cost function be exactly evaluated or solved. Such a requirement cannot be satisfied in our setting, because our cost function is more complex and generally intractable. In exchange for polynomial computation complexity, our algorithm replaces $\rho$ with $\rf$.
Recall that $\rf$ is a $(1+\theta)$-approximation to $\rho$. Now, the crux of analyzing our algorithm's competitiveness is to profile how differently Stage-1 will perform when it uses $\rf$ under a relaxed budget and $\rho$ under the original budget. 

Without loss of generality, we focus on Stage-1 of the $i$-th iteration of our algorithm. Suppose that we can, in parallel, run $\rho$ and $\rf$ in the cost evaluation (line \ref{algo:rtEvaluation} of Algorithm \ref{algo}), with $\bf{c}$ and $(1+\theta)\bf{c}$ as their budgets, respectively. For convenience, Stage-1's \while is called $\rho$-\while if $\rho$ is used in cost evaluation, and it is called $\rf$-\while if $\rf$ is used.

\begin{theorem}\label{thr:stop1}
\emph{
Given $S_i$ at the beginning of Stage-1, assume the $\rf$-\while stops with $X_i^\rf\subseteq S_i$. If $\rho$-\while can stop with $X_i^\rho\subseteq S_i$, then we always have $X_i^\rho\subseteq X_i^\rf$ for any $\theta>0$.
}
\end{theorem}
\begin{proof}
In line \ref{algo:cherryPick} of our algorithm, picking $s^*$ (i.e., the item with the highest marginal gain) out of $S_i$ has nothing to do with the solver used by the cost evaluation (line \ref{algo:rtEvaluation}). Without considering the routing constraint, the sequence of best vertices chosen in the \while  would be always the same. We denote this sequence by $\langle s_{1}, s_{2},\cdots s_{m}\rangle$, where $m=|S_i|$ and $s_{j}$ is picked ahead of $s_{k}$ if $j<k$.

Consider a general case: $\rf$-\while stops at $s_{k}~(k< m)$, while returning a route of cost $\rf_k$. As shown in Fig. \ref{fig:stopCondition}, in such a case, we have $\rf_{k}\leq (1+\theta)\bgt$~but~$\rf_{k+1}> (1+\theta)\bgt$. 
Now suppose that optimal solver $\rho$ can stop at $s_{k+1}$ with the best route of cost $\rho_{k+1}$,  meaning $\rho_{k+1}\leq\bgt$.

Since $\rf$ approximates $\rho$ with a factor of $(1+\theta)$, we readily have  
\begin{equation}\label{eqn:contradiction1}
	(1+\theta)\rho_{k+1}\geq \rf_{k+1} > (1+\theta)\bgt\,,
\end{equation}

\noindent which means $\rho_{k+1}>\bgt$. That contradicts. Hence,  $\rho$ cannot stop at $s_{k+1}$ under budget $\bgt$, if $\rf$ stops at $s_{k}$ under budget $(1+\theta)\bgt$. Clearly, $\rho$ cannot stop at any $s_j$ for $k+2\leq j\leq m$, too.  We thus have $X_i^\rho\subseteq X_i^\rf\subseteq S_i$ for any $\theta>0$.
\end{proof}

\begin{figure}\centering
\includegraphics[width=.6\textwidth]{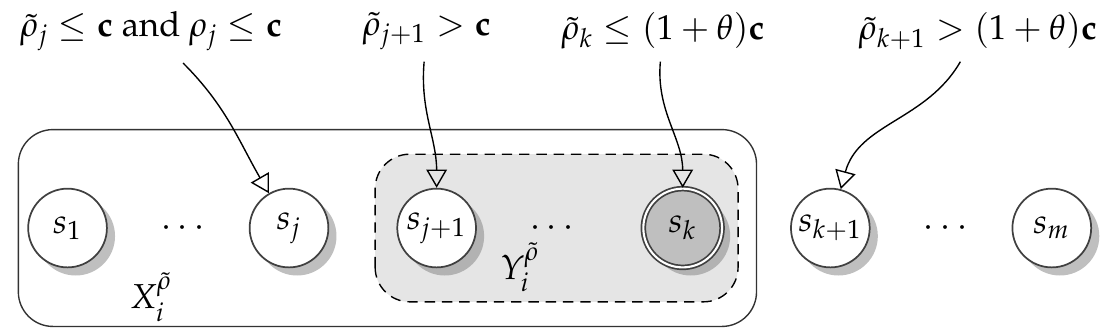}
\caption{Demo of our algorithm's Stage-1 selecting $X_i^\rf$ out of $S_i=\{s_1, s_2\ldots s_m\}$, where $j<k\leq m$.}
\label{fig:stopCondition}
\end{figure}

Theorem \ref{thr:stop1} shows that, in Stage-1, $\rho$ cannot pick more items out of $S_i$ under the budget of $\bgt$ than $\rf$ does under the budget of $(1+\theta)\bgt$. 

\begin{theorem}\label{thr:feasibleSolution}
\emph{
Given $S_i$ at the beginning of Stage-1, if $\rho$-\while stops with $X_i^\rho$, then our algorithm must be able to meet $X_i^\rho$ in the subsequent Stage-2, and determine a subset of $X_i^\rho$ as a candidate solution.
}
\end{theorem}
\begin{proof}
We resort to Fig. \ref{fig:stopCondition} to finish this proof. Assume that the $\rf$-\while stops in Stage-1 with $X_i^\rf$ and $Y_i^\rf$, and that $\rf_j\leq\bgt$ but $\rf_{j+1}>\bgt$. This means $Y_i^\rf=\langle s_{j+1}, s_{j+2}\ldots s_k\rangle$, according to our algorithm's lines \ref{algo:beyondBudget} and \ref{algo:Y}. 

Since $\rf$ approximates $\rho$ and $\rf_{j}\leq\bgt$, we have $\rho_j\leq\bgt$, i.e., that $\rho$ will not stop earlier than meeting item $j$. Additionally, we know $X_i^\rho\subseteq X_i^\rf$ by Theorem \ref{thr:stop1}. Thus $\rho$ must stop at some item of $\{s_j\}\cup Y_i^\rf$. 
If $X_i^\rho = X_i^\rf$, obviously, our algorithm can meet $X_i^\rho$ on Stage-2's line \ref{algo:xISopt}. 
Moreover, in Stage-2's line \ref{algo:last}, we continue to shrink $X_i^\rf$~(i.e., $X_i$), by deleting its last item, until $X_i^\rf\bs Y_i^\rf$ is hit. During such a shrinking process, we apply \unst to each of the intermediate subsets. Clearly, $X_i^\rho$ is either among these subsets, or equiverlant to $X_i^\rf$, i.e., 
our algorithm cannot miss $X_i^\rho$. 
\end{proof}

Theorem \ref{thr:stop1} and Theorem \ref{thr:feasibleSolution} have paved the way to prove our algorithm's bicriterion competitiveness. 
Next we introduce the formal definition of \emph{bicriterion approximation} as follows. 
For a problem of $\max\{f(X)|g(X)\leq c\}$, a $(p, q)$-bicriterion approximation algorithm can output a solution $X$ that guarantees $f(X)\geq p\cdot f(X^{\rm opt})$ and $g(X)\leq q\cdot c$, where $X^{\rm opt}$ is an optimal solution, $0<p<1$ and $q>1$. These two parameters $p$ and $q$ approximately measure the optimality and the feasibility of approximation algorithm, respectively.

\begin{theorem}\label{thr:competitive}
\emph{
Using a $(1+\theta)$-approximation algorithm $\rf$ to evaluate the cost constraint under budget $(1+\theta)\bgt$, our algorithm can at least achieve a $\left[\frac{k}{4(k+1)^2}, 1+\theta\right]$-bicriterion approximation.
}
\end{theorem}
\begin{proof}
In each iteration's Stage-1, our algorithm always stops adding items when its budget is violated. So, any routes in $\mathcal{T}$ are not beyond $(1+\theta)\bgt$ in total cost; in other words, the over-budget amount of our solution is at most $\theta\cdot\bgt$.

Theorem \ref{thr:feasibleSolution} ensures that our algorithm with $\rf$ in cost evaluation can still achieve an approximation ratio that is at least as good as SMS's $\frac{k}{(1+1/\alpha)(k+1)^2}$. 
Our algorithm's Stage-2 invokes the \unst procedure, whose approximation ratio is $0<\alpha<1$. It is proven in \cite{Buchbinder-2012-nonmonotone-submodular} that, for general unconstrained submodular maximization problems, there exist a deterministic and a randomized algorithms, which both need linear time and can achieve $1/3$-approximation and $1/2$-approximation (in expectation), respectively. If taking $\alpha=1/3$, our algorithm can achieve an exact approximation ratio of at least $\frac{k}{4(k+1)^2}\approx \frac{1}{4k}$. 
In \cite{Feige-2011-submodular,Buchbinder-2012-nonmonotone-submodular}, two elegant deterministic $1/3$-approximation \unst algorithms are presented. We in this paper use the \unst algorithm proposed in \cite{Buchbinder-2012-nonmonotone-submodular}, which is given in Appendix \ref{appdxA:deterministicUSM}. 
\end{proof}

\subsection{Asymptotic Time Complexity}\label{subsec:timeComplexity}

Without loss of generality, we first analyze the asymptotic time complexity in the $i$-th iteration with $S_i$ at the beginning. 

\textbf{Time cost in Stage-1}. The \while of Stage-1 will stop, if $S_i$ turns empty or if the relaxed constraint is violated; consequently, Stage-1 examines at most $|S_i|$ items. 
In Stage-1, determining all possible $s^*$ on line \ref{algo:cherryPick} needs $\Theta(|S_i|^2)$ calls of $f$-oracles in total.
In addition, it is easy to know that Stage-1 needs $\Theta(|S_i|)$ calls of $\rf$-oracles in constraint evaluation.

\textbf{Time cost in Stage-2}. In Stage-2, \unst is called  $|Y_i|$ times, and every call for \unst needs $\mathcal{O}(|X_i|)$ $f$-oracles. 
With $Y_i\subseteq X_i\subseteq S_i$, therefore, Stage-2's time cost can be loosely upper bounded by $\mathcal{O}(|S_i|^2)$ calls of $f$-oracle.

Besides involving the two-stage \for as the protagonist, our algorithm in its last step sorts through the collection $\mathcal{T}$ to return the best solution, and clearly, this process needs $\Theta(|\mathcal{T}|)$ time. Stage-2 of iteration $i$ adds into $\mathcal{T}$ at most $|Y_i|$ local optima; thus, the size of $\mathcal{T}$ is far less than $k|\Omega|$. 

To sum up, the time complexity of our algorithm is upper bounded by $\mathcal{O}(kn^2)$ $f$-oracles plus $\mathcal{O}(kn)$ $\rf$-oracles, where $n=|\Omega|$. We will in future give each stage a tighter analysis that can further reduce this time complexity. 

For the settings with $k$-system constraint, we cannot easily know $k$ in advance because determining $k$ often requires exponentially exhaustive computation. In literature, $k$ is usually replaced by the size of $\Omega$. Although doing so will not worsen the theoretical time complexity too much, less $k$ is preferable in practice. 
In \cite{AAGK-2010-nonsubmodular}, the authors recommend that $k$ can be set to two---running their SMS's iteration only twice instead of $|\Omega|$ times---in order to considerably reduce the run-time. 
In \cite{Feldman-2017-nonmonotone-submodular}, the authors find that, in the SMS-like algorithms, reducing $k$ down to $\sqrt{k}$ will not impact the algorithm approximation performance; that is,  in our algorithm's implementation, $k$ can be safely replaced by $\sqrt{|\Omega|}$, thereby leading to a significant reduction in computational time.

\subsection{Remark about the Over-budget}\label{subsec:control}

Our algorithm's solution may be beyond the original budget $\bgt$, and the amount of over-budget is at most $\theta\bgt$. Recall that the error parameter $\theta$ is the approximation factor of algorithm $\rf$ evaluating the routing cost. Obviously, the smaller $\theta$ is, the less the over-budget will be.

In many practical routing constraint settings, there exist $(1+\theta)$-approximation algorithms, where $\theta$ values are very small numbers, often less than one. 
For instance, there exists 1.55-approximation algorithm (i.e., $\theta=0.55$) for minimum-weight multicast routing problem~\cite{Du-2012-approx-bk}. 
If the routing cost is determined by Euclidean TSP, for instance, we can recruit a PTAS (Polynomial Time Approximation Scheme) for TSP problem and set its $\theta$ to be any small to control the over-budget below a very small scale. 
In detail, PTAS is a type of approximation algorithm for computation-intractable problems. A PTAS algorithm can achieve $(1+\theta)$-approximation for minimization, or $(1-\theta)$-approximation for maximization, where $\theta>0$; and for any fixed $\theta>0$, it runs in time polynomial in the input instance's size. For instance, there exists PTAS for Euclidean TSP problems, which finds a $(1+\theta)$-approximation solution in time of at most $\mathcal{O}(n(\log n)^{o(1/\theta)})$; there exists PTAS  for Subset-Sum problems, in time of $\mathcal{O}(n^3(\log n)^2/\theta)$.

In our algorithm, though a smaller $\theta$ of PTAS might lead to more computation time in Stage-1's cost evaluation, it can in some degree reduce the time cost in Stage-2 by reducing the size of $Y_i$. The following theorem gives an explanation about such a run-time reduction. 

\begin{theorem}\label{thr:stop2}
\emph{
Given $S_i$, if the $\rho$-\while stops with $X_i^\rho\subseteq S_i$ and the $\rf$-\while, with $X_i^\rf\subseteq S_i$, then we always have $|X_i^\rf \bs X_i^\rho|\leq 1$ for any feasible $0<\theta \leq \cmin/\bgt$, where $\theta$ is the error parameter of $\rf$. 
}
\end{theorem}
\begin{proof}
Since Theorem \ref{thr:stop1} has proven $X_i^\rho\subseteq X_i^\rf$ for any $\theta$, we here suppose that, given $S_i$, the $\rho$-\while can stop at $s_k$ if the $\rf$-\while stops at $s_{k+2}$. We next complete this proof by contradiction. 

From the conditions in which $\rho$ and $\rf$ stop, we have  (1) $\rf_{k+2}\leq(1+\theta)\bgt$, and (2) $\rho_{k}\leq\bgt$ but $\rho_{k+1}>\bgt$. According to the property of routing cost function, it is easy to know $\rho_{k+2}\geq\rho_{k+1}+c_{k+2}>\bgt+\cmin$, where $c_{k+2}$ is the visiting cost at item $s_{k+2}$. 
Because $\rf$ approximates $\rho$ and $\theta\leq\cmin/\bgt$, we have $\rf_{k+2}\geq\rho_{k+2}$, which further leads to
\begin{equation}\label{eqn:contradiction2}
\rf_{k+2}\geq\rho_{k+2}>\bgt+\cmin\geq\bgt+\theta\bgt\,.
\end{equation}

This inequality contradicts $\rf_{k+2}\leq(1+\theta)\bgt$. If $\rf$-\while  stops at $s_{k+2}$, therefore, $\rho$-\while cannot stop at $s_k$. Moreover, the $\rho$-\while cannot stop at $s_{j}$ for any $j< k$, too. Thus, $\rho$-\while can only stop either at $s_{k+1}$ or at $s_{k+2}$. This theorem holds true.
\end{proof}

If $\rf$ is a PTAS with $\theta\leq\cmin/\bgt$, then we will have $|Y_i|=|X_i^\rf\bs X_i^\rho|\leq 1$ at the end of Stage-1. It means that Stage-2 can finish with at most two calls of \unst. So a PTAS of $\rho$ with small $\theta$ can lower Stage-2's time cost. 
With Theorem \ref{thr:stop2}, we can readily have Corollary \ref{cor:theta}: the run-time and the budget control can be traded off according to practical settings and resource supply.

\begin{corollary}\label{cor:theta}
\emph{
Given $S_i$, if the $\rho$-\while stops with $X_i^\rho\subseteq S_i$ and the $\rf$-\while, with $X_i^\rf\subseteq S_i$, then we always have $|X_i^\rf \bs X_i^\rho|\leq r$ for any $0<\theta\leq r\cdot\cmin/\bgt$, where $r$ is a positive integer ranging in $[1, \bgt/\cmin)$. 
}
\end{corollary}

\section{Experiments}\label{sec:Experiment}

In this section we conduct numerical case studies to investigate our algorithm's performance. 
We present two baseline algorithms, \bsRand and \bsRMax. Baseline Rand continually selects items in a random way until the budget is violated.  
Like ours, baseline \bsRMax is also an iterated greedy algorithm, which adds the item $\arg\max\{f_X^{+}(x)/\rf_X^{+}(x)|x\in \Omega\}$ in each iteration, until it uses up the budget or its marginal profit turns negative. By maximizing a benefit-cost ratio, \bsRMax follows an intuitive idea: preferring to select the item that can bring higher marginal gain with less cost increase.

\subsection{Case 1: Diversified Task Offloading in Big Data}

This case study is about the Example-1 early mentioned in section \ref{sub:problem}. We disseminate a personalized movie recommendation task to a set of processors that are curated in a complete undirected graph; the edge weight represents the delay between processors. Each processor has a local movies' set, and the movies' meta data is foreknown to the user. 

\begin{figure}
\centering
\includegraphics[width=.4\textwidth]{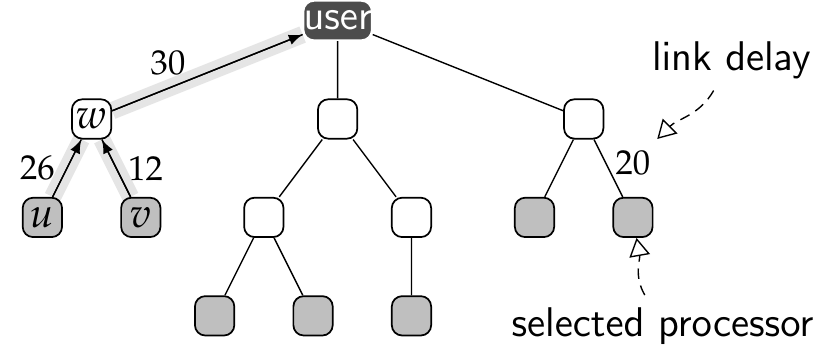}
\caption{Example of a multicast delay-aware routing tree.}
\label{fig:multicastTree}
\end{figure}

Fig. \ref{fig:multicastTree} shows a multicast routing tree that spans all the selected processors. Each processor executes the task and returns a result (in the form of packet), along the tree, to the user. We here assume that the links are all simplex and associated with a delay (i.e., edge weight).  In Fig. \ref{fig:multicastTree}, for example, processor $w$ first takes (26+12) time units to collect the packets from both $u$ and $v$, and then takes 30 time units to forward a combined (or compressed) packet to the user. In the multicast routing tree, therefore, the sum of all edge weights measures the total time for the user to retrieve the results from all the selected processors. 

Our goal in this example is to recommend a list of more representative movies while the response time stays under a certain time budget. 
Our dataset comes from MovieLens \cite{movielens} which involves the user information and their ratings of movies. However, the dataset is sparse in rating because most users rated only a handful of movies, i.e., the user-rating matrix is of low-rank. To address this issue, we used the PMF (probabilistic matrix factorization) model \cite{Salakhutdinov-2007-matrix} to generate a linearly-scaled rating matrix ${\rm M}_{k\times n}$ for $k$ users and $n$ movies.
We used the function in Eq. \eqref{diversity1} to calculate the objective value or the \emph{score} of movie recommendation, and set $\lambda$ to one. 
For two movies $i$ and $j$, we evaluated their similarity $s_{i,j}$ by using the inner product of their non-normalized feature vectors. 
In experiments, we assigned to each link a uniform random delay ranging from one to two hundred time units; and we applied the KMB~\cite{Kou-1981-steiner}, a 2-approximation algorithm, to producing a delay-minimum multicast tree over the processors selected in each iteration.

\begin{figure}[htbp]\centering
\subfigure[\sf ours]{
\includegraphics[width=0.20\textwidth]{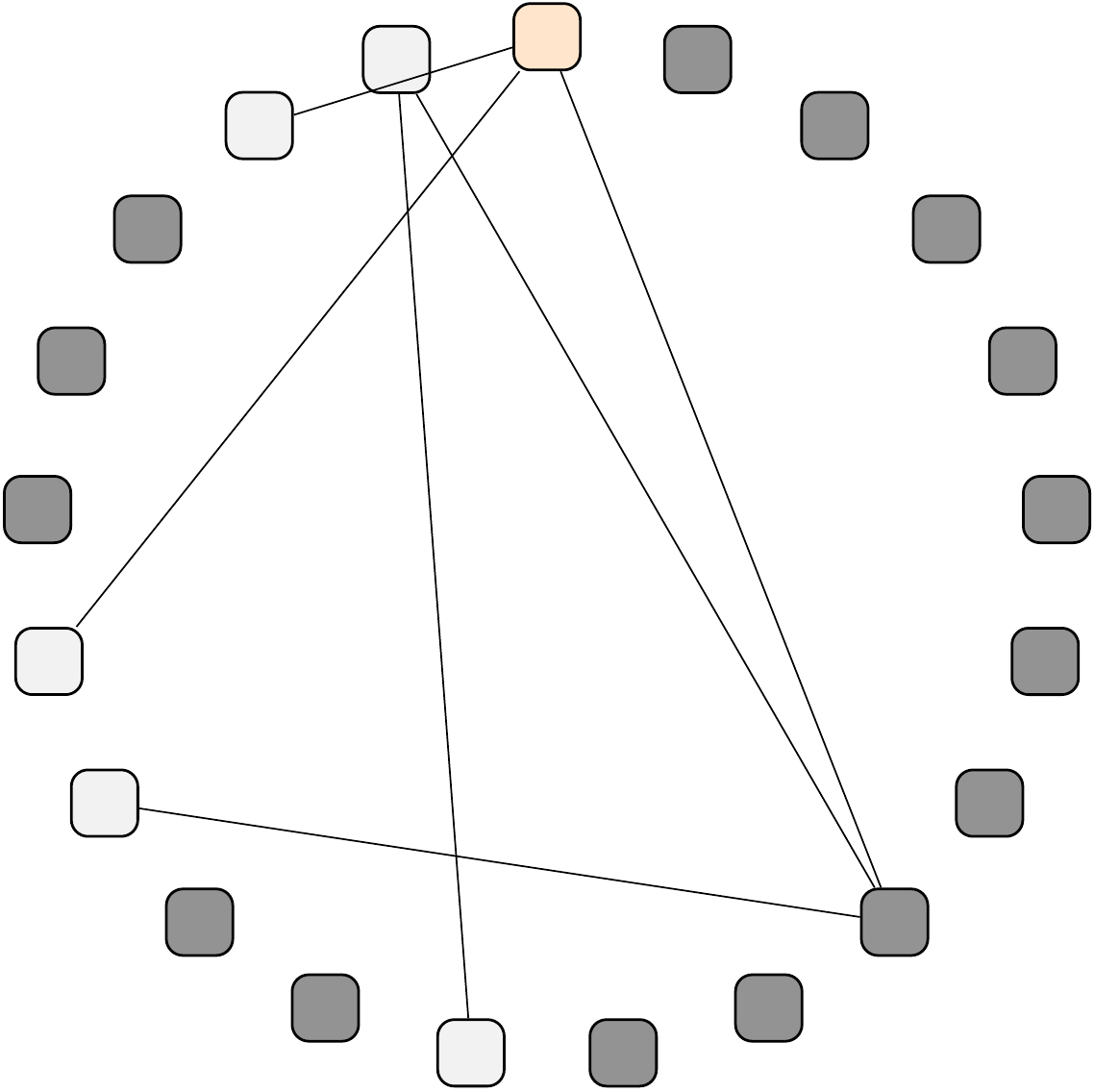}
\label{fig:graphOurs}
}\hfil
\subfigure[\sf rMax]{
\includegraphics[width=0.20\textwidth]{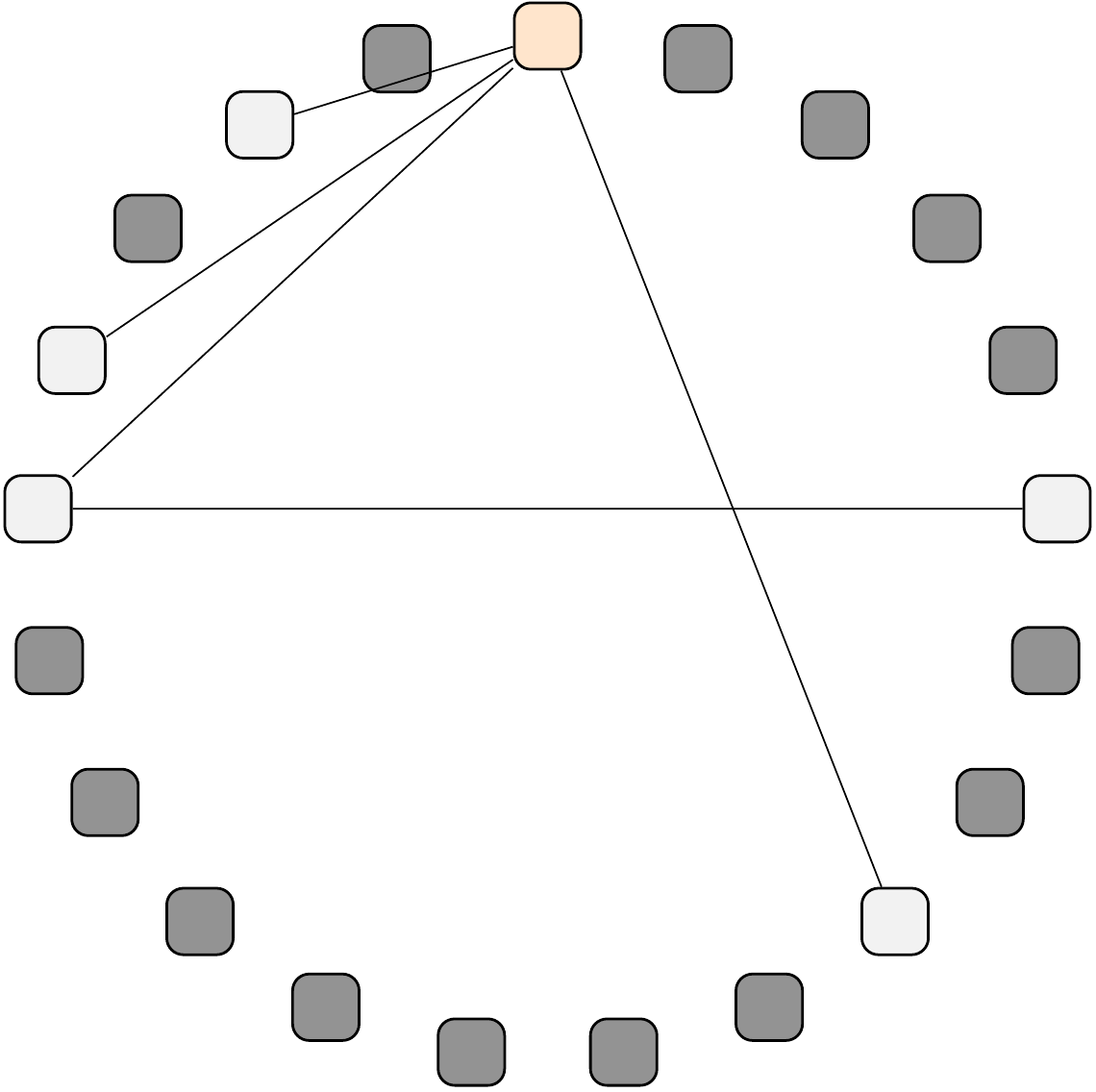}
\label{fig:graphrMax}
}\hfil
\subfigure[\sf Rand]{
\includegraphics[width=0.20\textwidth]{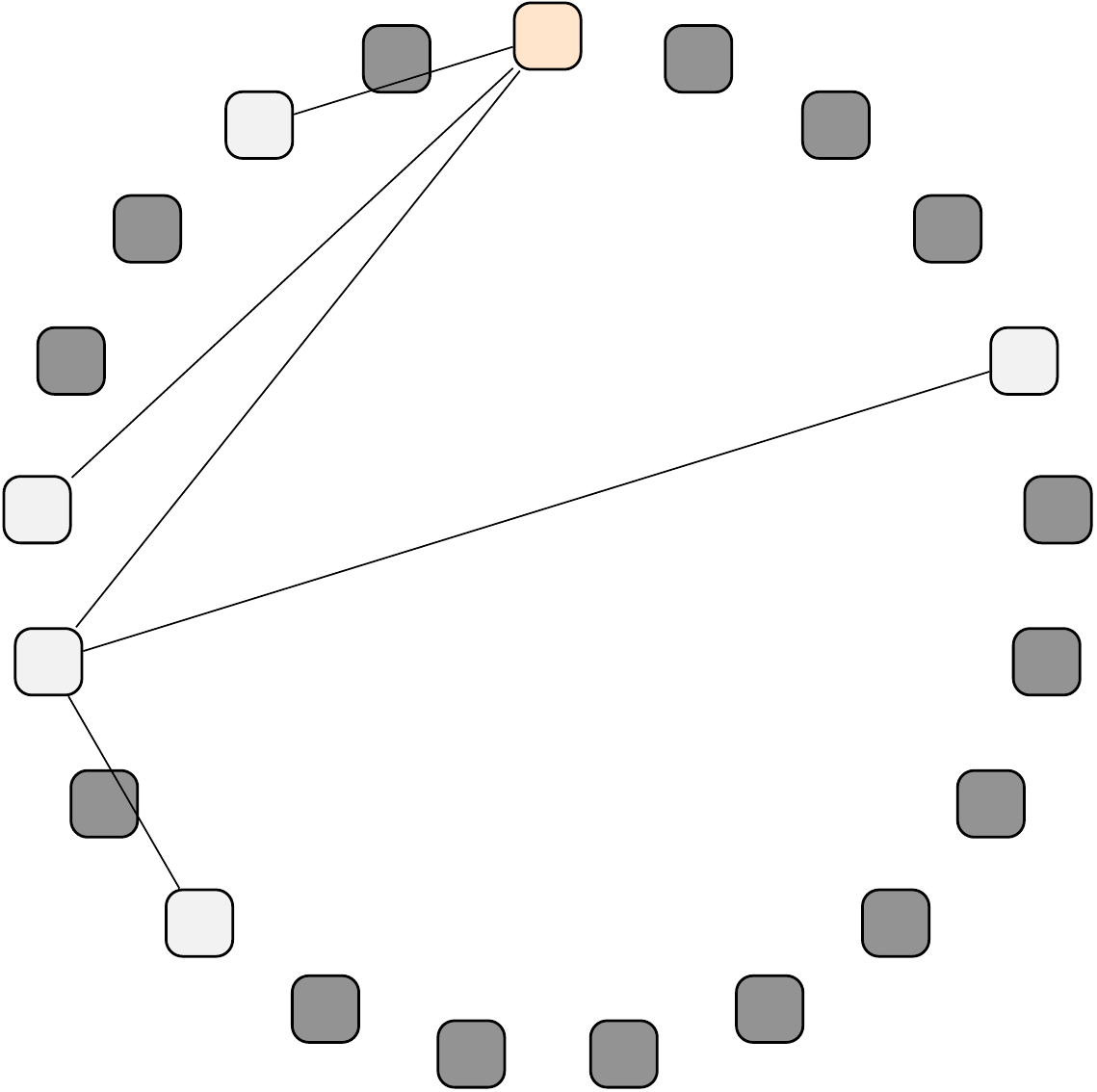}
\label{fig:graphRand}
}
\caption{Illustration of three algorithms in processor selection.}
\label{fig:multicastTrees}
\end{figure}

\begin{figure}\centering
\includegraphics[width=.685\textwidth]{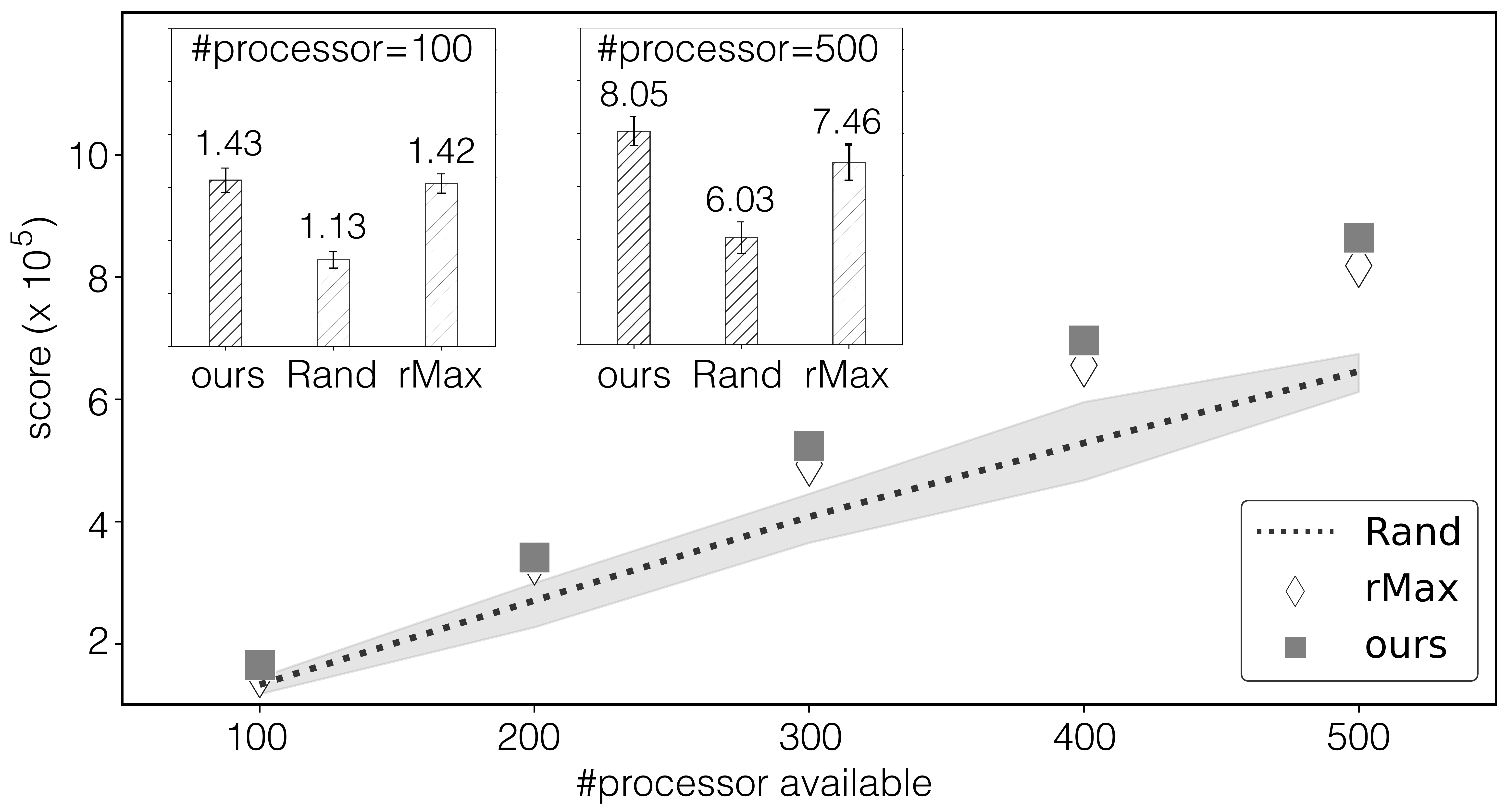}
\caption{Comparison of three algorithms under different network scales, where gray region is Rand's error band.}
\label{fig:scoreVsSize}
\end{figure}

Fig.~\ref{fig:multicastTrees} shows the processors selected by the  three algorithms to complete movie recommendation in the case where there are 20 available processors; for clarity, Fig.~\ref{fig:multicastTrees} only shows the result of a small-scale experiment. 
Fig.~\ref{fig:scoreVsSize} plots the effect of network scale on recommendation performance. For the three algorithms, an intuitive result is that the more processors (i.e., the more movies) are available, the higher the scores. For instance, with 100 processors, our algorithm can achieve a recommendation score of $1.43\times 10^5$, while, with 500 processors, its recommendation score is $8\times 10^5$, almost five times the score under 100 processors.
Such a comparison shows that more processors can bring a larger search space and our algorithm can find a better solution.
Of these three algorithms, \bsRand's recommendation scores the lowest for the two network scales. From the two sub-plots of Fig.~\ref{fig:scoreVsSize}, we can see that, for the 100-processor case, our algorithm's score is 0.7\% higher than \bsRMax's; however, for the 500-processor case, the performance enhancement stretches to 6.7\%.

\begin{figure}\centering
\subfigure[\sf \bsRand]{
\includegraphics[width=0.20\textwidth]{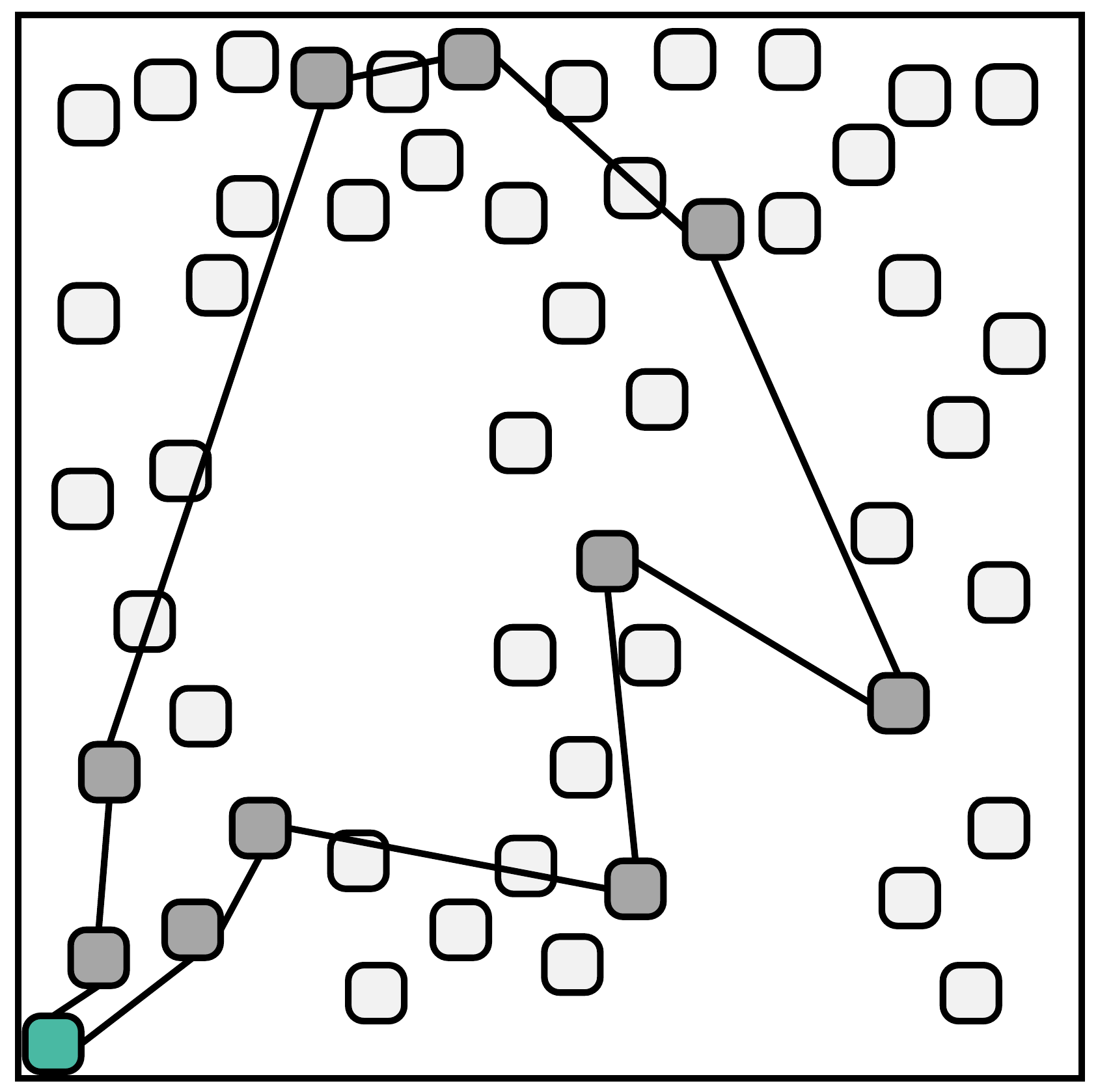}
\label{fig:senseRand}
}\hfil
\subfigure[\sf \bsRMax]{
\includegraphics[width=0.20\textwidth]{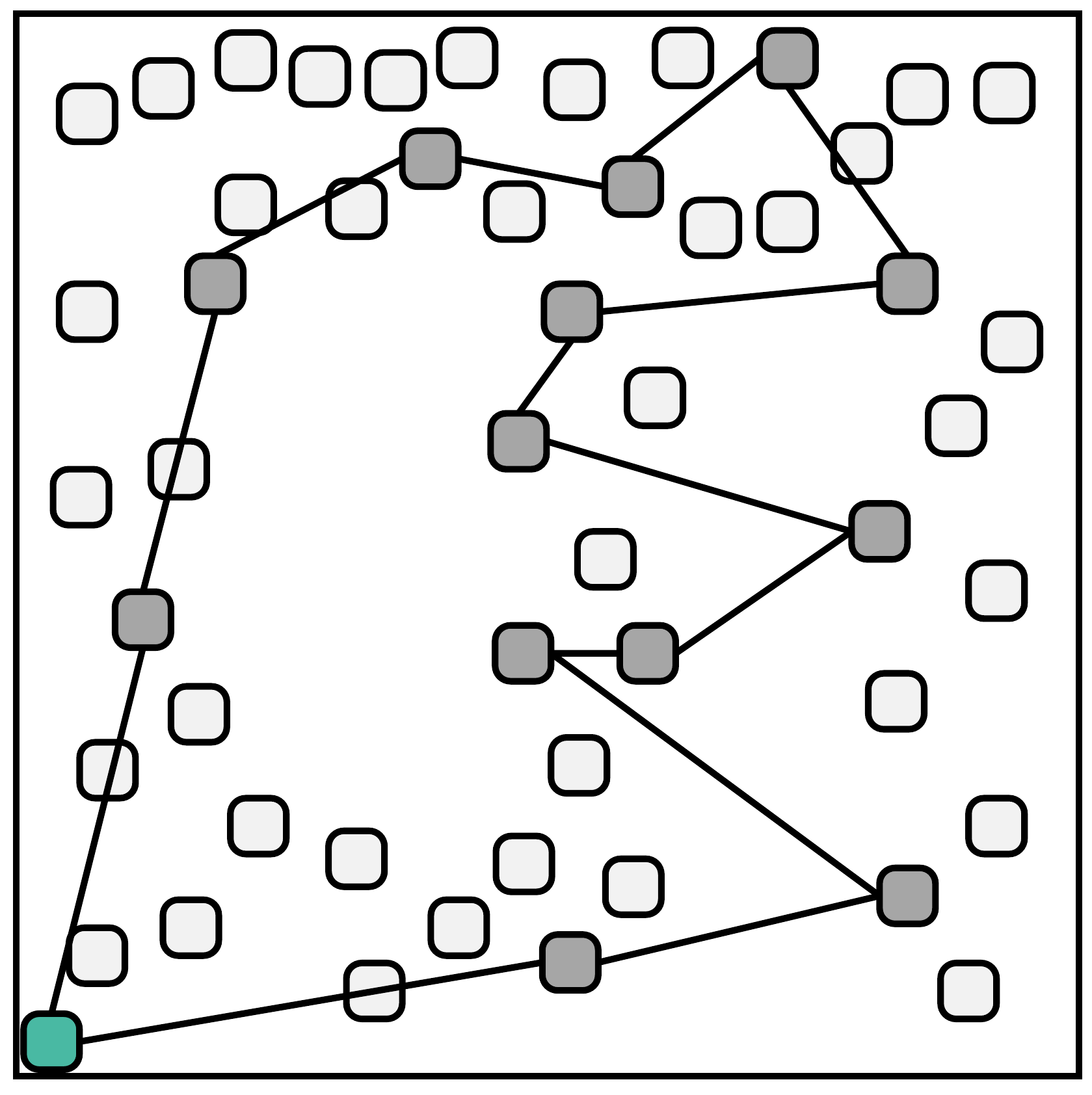}
\label{fig:senseRMax}
}\hfil
\subfigure[\sf ours]{
\includegraphics[width=0.20\textwidth]{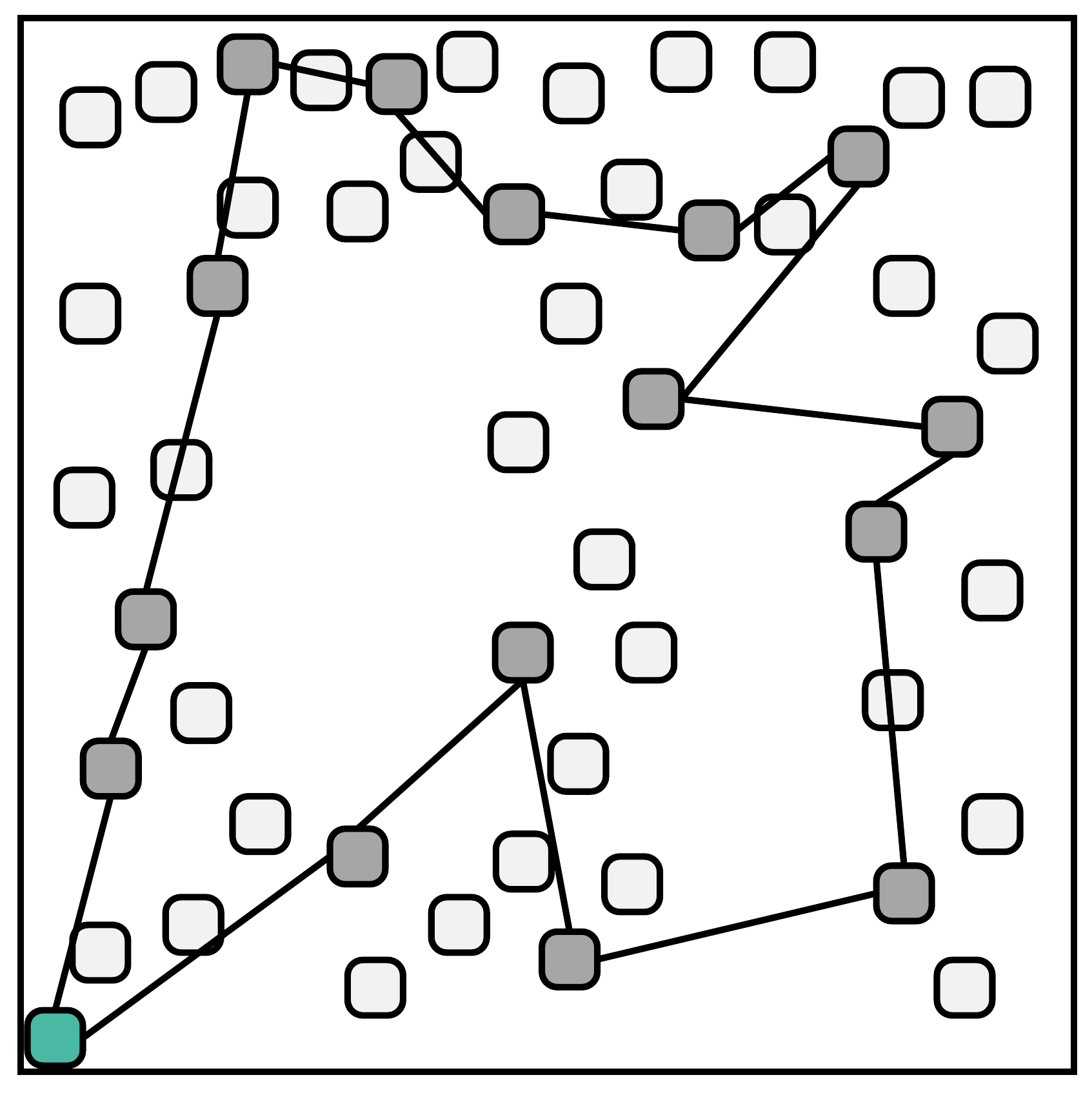}
\label{fig:senseOurs}
}
\caption{Illustration of three algorithms in the selection of PoI's (gray squares) and the robotic data-collecting trajectory.}
\label{fig:sensing}
\end{figure}

\subsection{Case 2: Informative Robotic Data Collection}

We consider a scenario of Example-2, where 45 points of interest (PoI's) are specified within a $35\times 40$ square monitoring area. A  charging depot is located at the corner of the monitoring area, and a robot departs from the depot to visit the PoI's to collect spatial data. The robot consumes energy in travel and data collection, and it must return to the depot for energy replenishment before it runs out of energy. 
We assume that the traveling-energy rate is 0.6  (i.e., it needs to pay $0.6d$ energy units for moving a distance of $d$), and the data-collecting energy at some PoI is a random number within $(0,2)$ that is specified in advance. 
We here use mutual information to measure how informative a subset of PoI's is, and create the spatial correlation between all the PoI's, from the data set obtained in a testbed experiment~\cite{Krause-2008-submodular-gaussian}. This testbed deployed 54 sensors to monitor the environmental temperature and humidity in an area of $35\times 40\mbox{m}^2$. 
To determine robot's energy-least TSP tour, we use the PTAS algorithm proposed in \cite{Mitchell-2007-TSPN} and set its error parameter $\theta$ to 0.1 in our experiments.

In this case, our goal is to find a data-collection schedule for the robot, such that the mutual information can be maximized while its energy budget is not violated. Appendix \ref{appdxC:MutualInfo} calculates the mutual information on the basis of Gaussian process model. 
Fig. \ref{fig:sensing} compares our algorithm and the two baselines, in which the robot's energy budget is set to 120 energy units. The PoI's selected by our algorithm are the most evenly in terms of spatial distribution. 

Fig. \ref{fig:gain} plots how the three algorithms will perform as the energy budget, and how their information gains change as the number of selected PoI's increases. 
An abundant energy supply usually leads to higher information gain, but it is not a panacea: for instance, \bsRMax's information gain can no longer be enhanced when the energy budgets reach 240. In particular, although \bsRand recklessly tries to use up budget to observe as many PoI's as possible, its information gain turns down as the budget goes higher than 320 energy units, because the mutual information measure is nonmonotone (i.e., more PoI's could lead to lower mutual information values). 

\begin{figure}\centering
\subfigure[under different budgets]{
\includegraphics[width=0.35\textwidth]{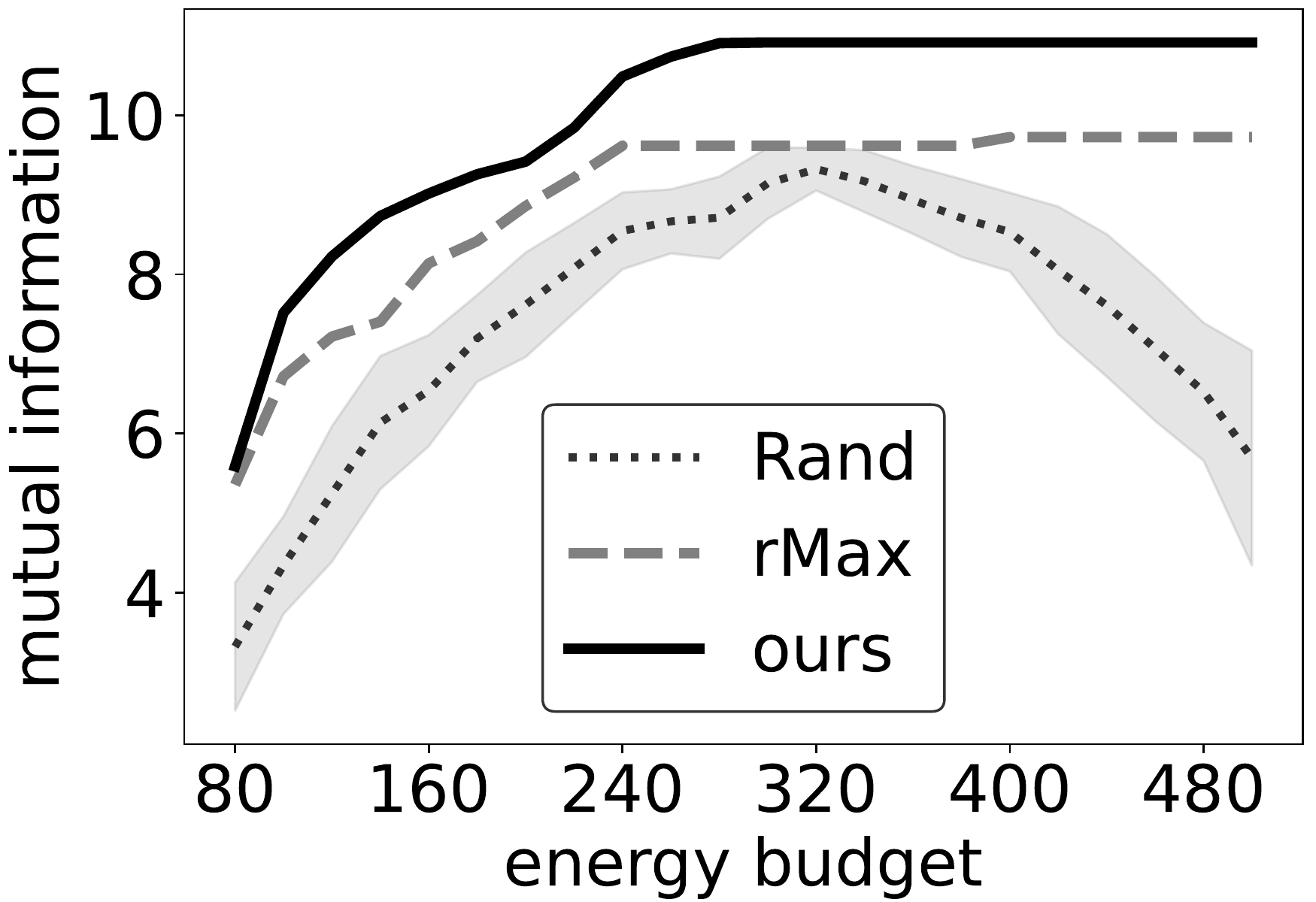}
\label{fig:MIvsBudget}
}\hfil
\subfigure[with different PoI counts]{
\includegraphics[width=0.35\textwidth]{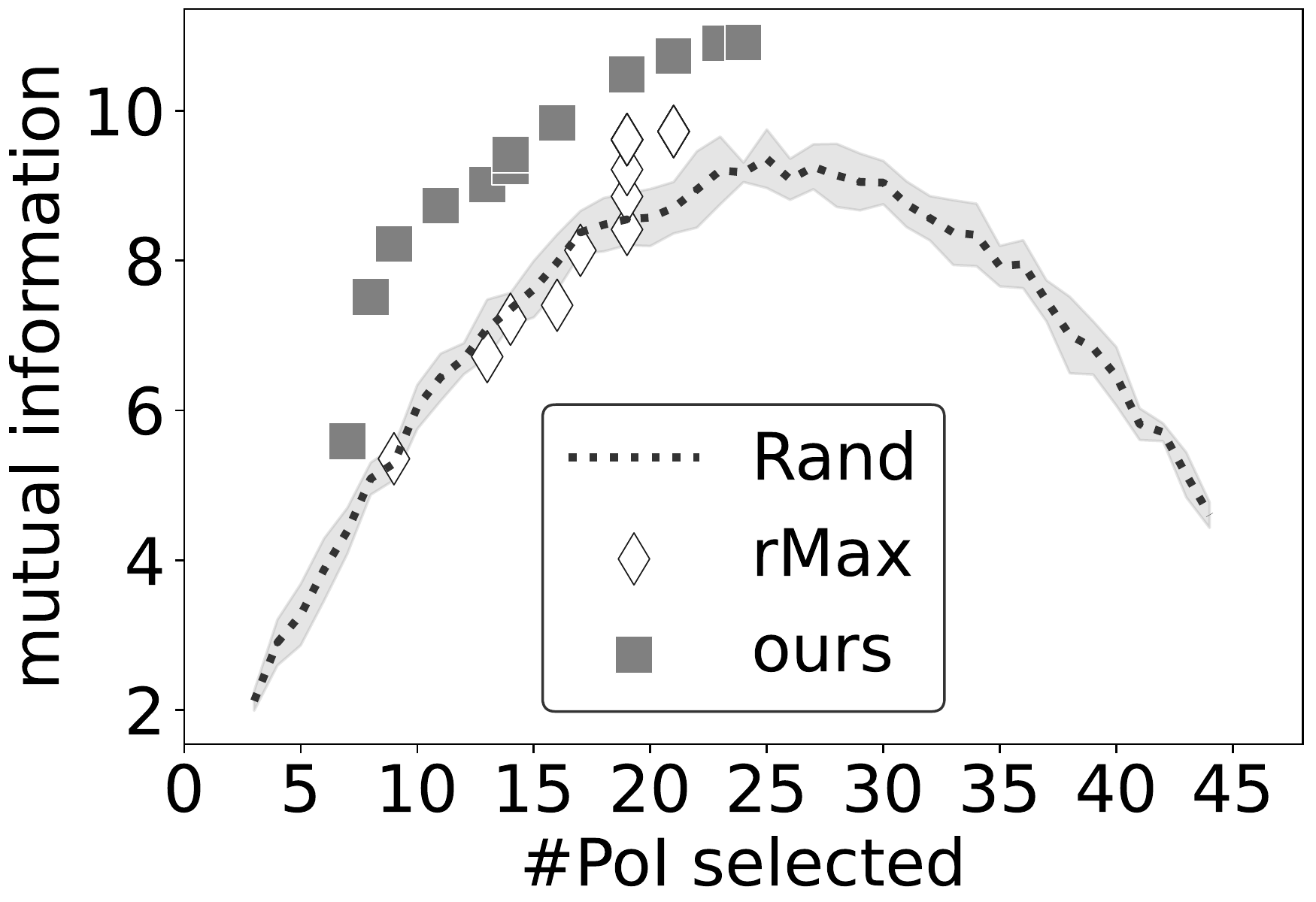}
\label{fig:MIvsNum}
}
\caption{Comparison of three algorithms in information, where gray regions are Rand's error band.}
\label{fig:gain}
\end{figure}

Fig. \ref{fig:energyAlloc} illustrates the energy allocation of three algorithms under repeated experiments with random seeds. In the case of budget=100, our algorithm pays 38.5\% of, and 51.9\% of, the energy budget for collecting data and for traveling, respectively, while \bsRand and \bsRMax  pay less energy for collecting data. On average, our algorithm visits about seven PoI's, while the two baselines can visit less-than-four PoI's. 
On the other hand, the total energy consumption rates are 90\%, 78\%, and 68\% for our algorithm, \bsRand, and \bsRMax, respectively. 
Similar comparison can also be found in the case of budget=200. Fig. \ref{fig:energyAlloc} shows that our algorithm can exploit the limited energy budget more effectively than the two baselines.

\begin{figure}\centering
\subfigure[budget=100]{
\includegraphics[width=0.375\textwidth]{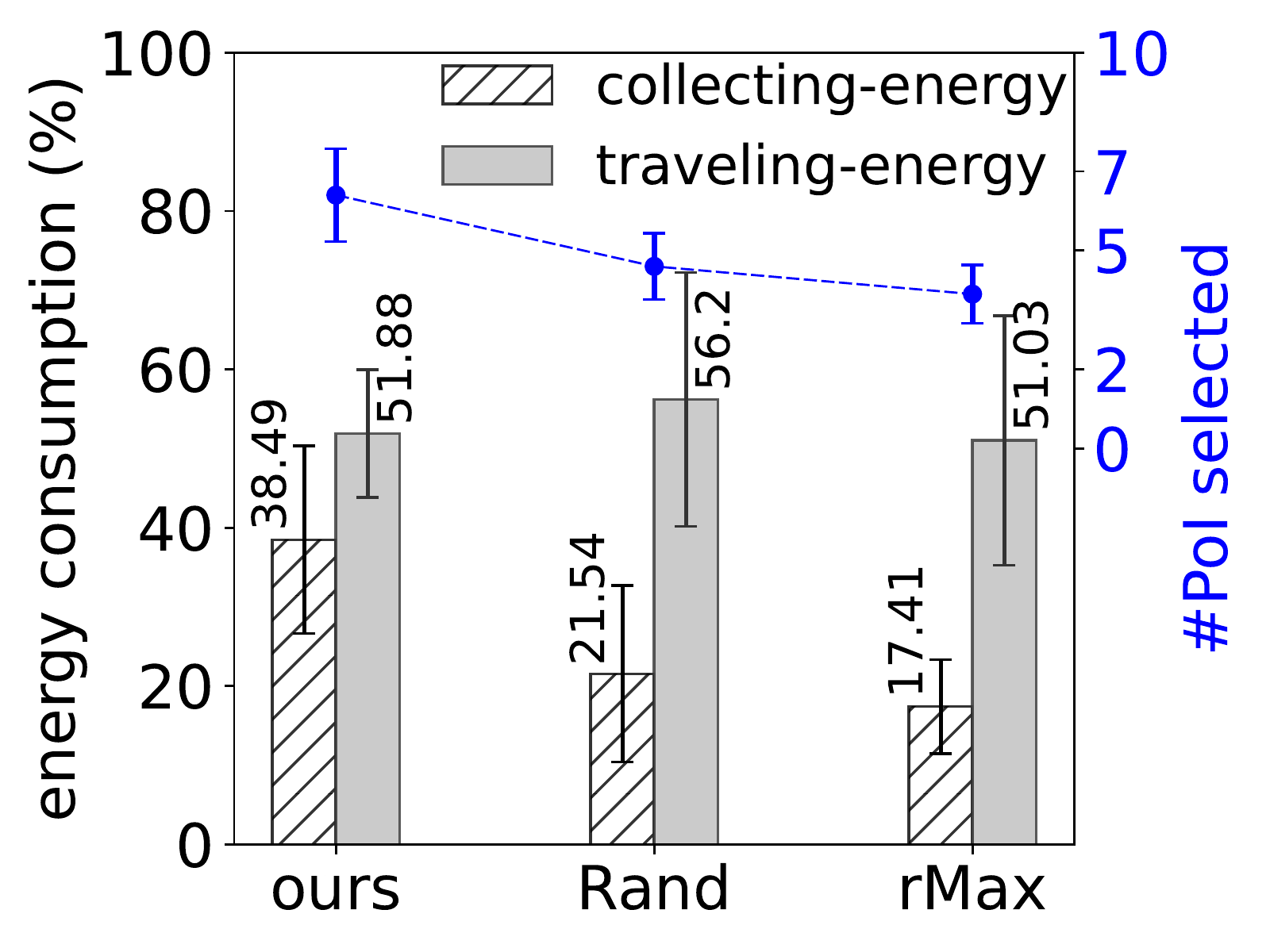}
\label{fig:budget100}
}\hfil
\subfigure[budget=200]{
\includegraphics[width=0.375\textwidth]{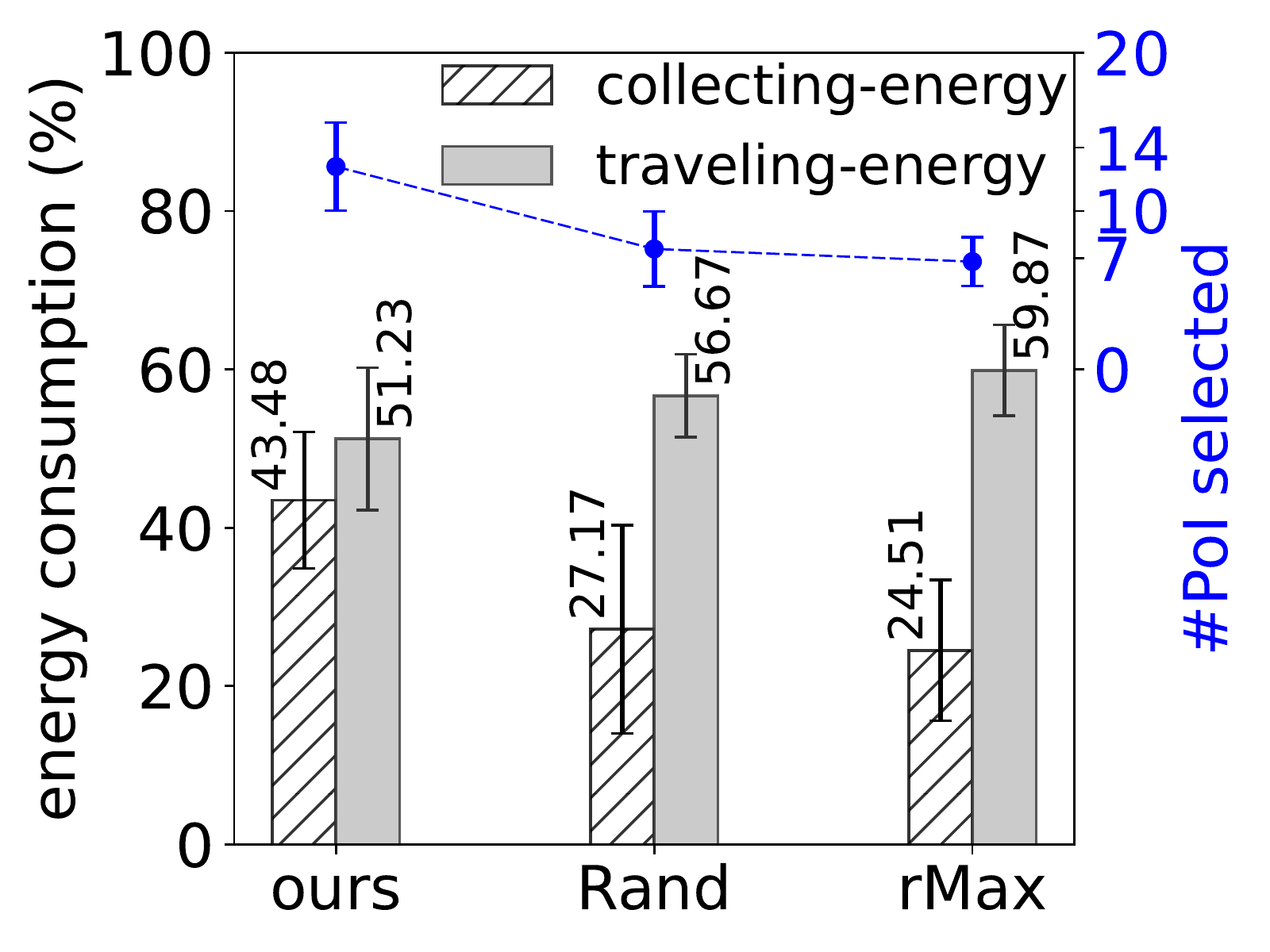}
\label{fig:budget200}
}
\caption{Comparison of three algorithms in energy efficiency.}
\label{fig:energyAlloc}
\end{figure}

\section{Conclusion}\label{sec:Conclustion}

In this paper, we have investigated the nonmonotone submodular maximization problem with routing constraint,  proved it to be a $k$-system-constrained nonmonotone submodular maximization problem, and proposed a $[1/4k, 1+\theta]$-bicriterion approximation algorithm. Considering the intractability of constraint evaluation, we draw an elegant connection between the constraint evaluation's error parameter (i.e., $\theta$) and our algorithm's overall performance. Further, we have developed the machinery for proving the performance guarantee; we believe it can shed light on maximizing nonmonotone submodular function whose constraint evaluation is generally NP-hard.
In addition, we have applied our algorithm to two motivating examples. The experimental results demonstrate our two-stage algorithm's effectiveness.





\appendix

\renewcommand{\theequation}{\thesection.\arabic{equation}}

\section{The deterministicUSM algorithm}\label{appdxA:deterministicUSM}

The \textsf{deterministicUSM} algorithm is designed in \citep{Buchbinder-2012-nonmonotone-submodular} to maximize nonmonotone submodular function without constraints. This algorithm is described in Algorithm \ref{USM} (with a few notation modifications from the original).

\SetAlgorithmName{Algorithm}{}{}
\begin{algorithm}
\caption{\textsf{deterministicUSM}}\label{USM} 
\DontPrintSemicolon
\SetKwInOut{Input}{input}\SetKwInOut{Output}{output}
\SetKwFor{While}{while}{}{end}
\SetKwFor{For}{for}{}{end}
\Input{$\Omega$ (items in some arbitrary order)}
\Output{$X_n\subseteq\Omega$ and $f(X_n)$}
\BlankLine
$X_0 \leftarrow \emptyset$ and $Y_0\leftarrow\Omega$ \;
\For{$i=1$~\emph{to}~$|\Omega|${~~{\tcp*[h]{$u_i$: the $i$-th item of $\Omega$}}}}{ 
	$a_i \leftarrow f(X_{i-1}+u_i) - f(X_{i-1})$ \;
	$b_i \leftarrow f(Y_{i-1}\bs u_i) - f(Y_{i-1})$ \;
	\eIf{$a_i\geq b_i$}{
		$X_{i}\leftarrow X_{i-1} + u_i$ and $Y_{i}\leftarrow Y_{i-1}$ \;
	}{
		$X_{i}\leftarrow X_{i-1}$ and $Y_{i}\leftarrow Y_{i-1}\bs u_i$ \;
	}
}
\Return $X_n$
\end{algorithm}

\section{Calculating Data Diversity}\label{appdxB:similarity}

Data diversity or data representativeness is commonly modeled as set functions that are nonnegative, submodular, and nonmonotone. 
For example, the movie recommendation problem is investigated in  ~\cite{Mirzasoleiman-2016-submodular}; and over a movie set $\Omega$, the following function is defined to evaluate the representativeness of a subset $S\subseteq\Omega$, or to give a score for $S$:

\begin{equation}\label{diversity1}
f(S)=\sum_{i\in\Omega\bs S}\sum_{j\in S}s_{i,j}-\lambda\sum_{i\in S}\sum_{j\in S}s_{i,j}
\end{equation}

\noindent where $s_{i,j}$ measures the similarity between the two movies $i$ and $j$, and $0<\lambda\leq 1$. When $\lambda=1$, the above function is the cut-function. This utility function is nonnegative and nonmonotone.

Similarly, another objective function is proposed in~\cite{Mirzasoleiman-2016-submodular} for the personalized image summarization problem that aims at finding the most representative images. This function is defined as 

\begin{equation}\label{diversity2}
f(S)=\sum_{i\in\Omega\backslash S}\,\max_{j\in S}\{d_{i,j}\}-\frac{1}{|\Omega\backslash S|}\sum_{i\in S}\sum_{j\in S}d_{i,j}
\end{equation}

\noindent where $d_{i,j}$ measures the similarity between the two images $i$ and $j$. Such an objective function is also nonmonotone submodular. 

\section{Calculating Mutual Information}\label{appdxC:MutualInfo}

Given $S\subseteq\Omega$, the mutual information of $S$ is expressed with 

\begin{eqnarray}\label{eqn:MI}
I(S; \Omega\bs S) & = &  h(\Omega\bs S) - h(\Omega\bs S|S) \notag\\
& = & h(S) + h(\Omega\bs S) - h(\Omega)\,,
\end{eqnarray}

\noindent where $h(\cdot)$ is the entropy function.   
Gaussian process is an excellent model of quantifying the uncertainty about prediction. A multivariate Gaussian process is usually featurized by a covariance matrix, $\mathbf{B}$, which is a kind of priors about all variables. 
In general, the joint entropy of any $k$-variate gaussian $S$ can be given by

\begin{equation}
h(S) = \frac{1}{2}\log\left[(2\pi e)^{k}\cdot \det(\mathbf{B}_S)\right]
\end{equation}

\noindent where $\mathbf{B}_S$ is the submatrix of $\mathbf{B}$, with the rows and columns both indexed by $S$.

In the Case-2 example, as the objective function, $I(S;\Omega\bs S)$ quantifies the expected reduction of uncertainty about the unobserved PoI’s upon the revelation of $S$.

\end{document}